\documentclass[sigplan,screen]{acmart}

\AtBeginDocument{%
  \providecommand\BibTeX{{%
    \normalfont B\kern-0.5em{\scshape i\kern-0.25em b}\kern-0.8em\TeX}}}

\usepackage{microtype}
\usepackage{graphicx}
\usepackage{latexsym}
\usepackage[utf8]{inputenc}
\usepackage{booktabs}
\usepackage{arydshln}
\usepackage{amssymb}
\usepackage{enumerate}
\usepackage{graphicx} 
\usepackage{amssymb,latexsym,amsmath,xspace}
\usepackage{enumerate,verbatim}
\usepackage{rotating}
\usepackage{environ}
\usepackage{multicol,esvect,relsize}
\usepackage{tikz,hyperref}

\newcommand{\classFont}[1]{\protect\ensuremath{\mathsf{#1}}\xspace}

\newcommand{\pind}{\perp\!\!\!\perp}
\newcommand{\cpind}{\perp\!\!\!\perp_{\rm c}}

\newcommand {\pci}[3] {#2~\!\!\pind_{#1}\!\!~#3}
\newcommand {\pmi}[2] {#1~\!\!\pind\!\!~#2}

\newcommand{\SUM}{\mathrm{SUM}}

\newcommand{\pcixyz}{\pci{\tuple x}{\tuple y}{\tuple z}}
\newcommand{\pmixy}{\pmi{\tuple x}{\tuple y}}

\newcommand{\PSPACE}{\classFont{PSPACE}}

\newcommand{\NP}{\classFont{NP}}
\newcommand{\bp}[1]{\mathrm{BP}(#1)}
\newcommand{\PNP}[2]{\boole{\classFont{NP}}{#1}{#2}}
\newcommand{\NPr}{\classFont{NP}_\mathbb{R}}
\newcommand{\NPrz}{\classFont{NP}^0_\mathbb{R}}
\newcommand{\PTIME}{\classFont{P}}
\newcommand{\PTIMEr}{\classFont{P}_\mathbb{R}}
\newcommand{\PTIMErz}{\classFont{P}^0_\mathbb{R}}

\newcommand{\boole}[3]{\mathrm{S}\textrm{-}#1^{#2}_{#3}}

\newcommand{\tuple}[1]{\vec{#1}}

\newcommand{\Dom}{\operatorname{Dom}}

\newcommand{\ar}{\operatorname{ar}}

\newcommand{\A}{\mathfrak{A}}

\newcommand{\N}{\mathbb{N}}

\newcommand{\calL}{\mathcal{L}}
\newcommand{\calS}{\mathcal{S}}
\newcommand{\calC}{\mathcal{C}}

\newcommand{\X}{\mathbb{X}}

\newcommand{\struc}{\mathrm{Struc}}

\newcommand{\sub}{\subseteq}

\newcommand{\Varfo}{\mathrm{Var_{1}}}
\newcommand{\Varso}{\mathrm{Var_{2}}}

\newcommand{\dfn}{:=}

\newcommand{\fleq}{\leq_{\mathrm{fin}}}
\newcommand{\enc}{\mathrm{enc}}

\newcommand{\RE}{\mathbb{R}}

\newcommand{\FO}{{\rm FO}}

\newcommand{\mA}{{\mathfrak A}}

\newcommand{\ESO}{{\rm ESO}}
\newcommand{\eso}[2]{\ESO_{#1}[{#2}]}
\newcommand{\peso}[2]{\mathrm{L}\text{-}\ESO_{#1}[{#2}]}

\newcommand{\esor}[1]{\eso{\RE}{#1}}

\newcommand{\esod}[1]{\eso{d[0,1]}{#1}}
\newcommand{\pesod}[1]{\peso{d[0,1]}{#1}}
\newcommand{\pesou}[1]{\peso{[0,1]}{#1}}

\newcommand{\ESOr}{\ESO_{\mathbb{R}}}

\newcommand\bigexists{%
  \mathop{\lower0.75ex\hbox{\ensuremath{%
    \mathlarger{\mathlarger{\mathlarger{\mathlarger{\exists}}}}}}}%
  \limits}

\theoremstyle{theorem}

\newcommand\copyrighttext{%
  \footnotesize  \textcopyright Authors 2020. This is the author's version of the work. It is posted here for your personal use. Not for redistribution. The definitive version was published in Proceedings of the 35th Annual ACM/IEEE Symposium on Logic in Computer Science (LICS '20), July 8--11, 2020, Saarbr\"ucken, Germany, ACM, New York, NY, USA, https://doi.org/10.1145/3373718.3394773.}
\newcommand\copyrightnotice{%
\begin{tikzpicture}[remember picture,overlay]
\node[anchor=north,yshift=-10pt] at (current page.north) {\fbox{\parbox{\dimexpr\textwidth-\fboxsep-\fboxrule\relax}{\copyrighttext}}};
\end{tikzpicture}%
}

\copyrightyear{2020} 
\acmYear{2020} 
\setcopyright{acmlicensed}
\acmConference[LICS '20]{Proceedings of the 35th Annual ACM/IEEE Symposium on Logic in Computer Science}{July 8--11, 2020}{Saarbr\"ucken, Germany}
\acmBooktitle{Proceedings of the 35th Annual ACM/IEEE Symposium on Logic in Computer Science (LICS '20), July 8--11, 2020, Saarbr\"ucken, Germany}
\acmPrice{15.00}
\acmDOI{10.1145/3373718.3394773}
\acmISBN{978-1-4503-7104-9/20/07}

\usepackage{booktabs}   
\usepackage{subcaption} 

\begin{document}

\title[Real computation and probabilistic team semantics]{Descriptive complexity of real computation and probabilistic independence logic}


\author{Miika Hannula}
\orcid{0000-0002-9637-6664}             
\affiliation{
  \institution{University of Helsinki}            
  \country{Finland}                    
}
\email{miika.hannula@helsinki.fi}          

\author{Juha Kontinen}
\orcid{[0000-0003-0115-5154}             
\affiliation{
	\institution{University of Helsinki}            
	\country{Finland}                    
}
\email{juha.kontinen@helsinki.fi}          

\author{Jan Van den Bussche}
\orcid{0000-0003-0072-3252}             
\affiliation{
	\institution{Hasselt University}            
	\country{Belgium}                    
}
\email{jan.vandenbussche@uhasselt.be}          

\author{Jonni Virtema}
\orcid{0000-0002-1582-3718}             
\affiliation{
	\institution{Hokkaido University}            
	\country{Japan}                    
}
\email{jonni.virtema@let.hokudai.ac.jp}          
\affiliation{
	\institution{Hasselt University}            
	\country{Belgium}                    
}

\renewcommand{\shortauthors}{M. Hannula, J. Kontinen, J. Van den Bussche, and J. Virtema}

\begin{abstract}
	
We introduce a novel variant of BSS machines called Separate Branching BSS machines (S-BSS in short) and develop a Fagin-type logical characterisation for languages decidable in nondeterministic polynomial time by S-BSS machines.  We show that NP on S-BSS machines is strictly included in NP on BSS machines and that every NP language on S-BSS machines is a countable disjoint union of closed sets in the usual topology of $\RE^n$. Moreover, we establish that on Boolean inputs NP on S-BSS machines without real constants characterises a natural fragment of the complexity class $\exists \RE$ (a class of problems polynomial time reducible to the true existential theory of the reals)  and hence lies between $\NP$ and $\PSPACE$. Finally we apply our results to determine the data complexity of probabilistic independence logic.

\end{abstract}

\ccsdesc[500]{Theory of computation~Complexity theory and logic}
\ccsdesc[500]{Theory of computation~Finite Model Theory}
\ccsdesc[500]{Mathematics of computing~Probability and statistics}
\ccsdesc[500]{Theory of computation~Models of computation}

\keywords{Blum-Shub-Smale machines, descriptive complexity, team semantics, independence logic, real arithmetic.}  

\maketitle

\copyrightnotice

 \section{Introduction}
 \emph{The existential theory of the reals} consists of all
 first-order sentences that are true about the reals and are of the form
 \[
 \exists x_1\ldots \exists x_n\phi(x_1, \ldots ,x_n),
 \]
 where $\phi$ is a quantifier-free arithmetic formula containing inequalities and equalities. Known to be $\NP$-hard on the one hand, and in $\PSPACE$ on the other hand \cite{Canny88}, the exact complexity of this theory is a major open question.  
  The existential theory of the reals is today attracting considerable interest due to its central role in geometric graph theory. First isolated as a complexity class in its own right in \cite{Schaefer09}, $\exists \RE$ is defined as the closure of the existential theory of the reals under polynomial-time reductions. In the past decade several algebraic and geometric problems have been classified as complete for $\exists \RE$; a recent example is the art gallery problem of deciding whether a polygon can be guarded by a given number of guards \cite{AbrahamsenAM18}. 
 
 The existential theory of the reals is closely connected to
 \emph{Blum-Shub-Smale machines} (BSS machine for short) which
 are essentially random access machines with registers that can
 store arbitrary real numbers and which can compute rational
 functions over reals in a single time step. Many complexity classes from
 classical complexity theory transfer to the realm of BSS machines, such as
 nondeterministic polynomial time ($\NPr$) over languages
 consisting of finite strings of reals. While the focus is
 primarily on languages over some numerical domain (e.g., reals
 or complex numbers), also Boolean inputs (strings over
 $\{0,1\}$) can be considered. In this context $\exists \RE$ corresponds to the \emph{Boolean part} of $\NPr^0$ ($\bp{\NPr^0}$), obtained by restricting $\NP_\mathbb{R}$ to Boolean inputs and limiting the use of machine constants to $0$ and $1$, as feasibility of Boolean combinations of polynomial equations is complete for both of these classes \cite{BurgisserC06,SchaeferS17}. 
 
 BSS computations can also be described logically. This research orientation was initiated by Gr\"adel and Meer who showed that $\NPr$ is captured by a variant of existential second-order logic ($\ESO_{\mathbb{R}}$) over \emph{metafinite structures} \cite{GradelM95}. Metafinite structures are two-sorted structures that consist of a finite structure, an infinite domain with some arithmetics (such as the reals with multiplication and addition), and weight functions bridging the two sorts \cite{GradelG98}. Since the work by Gr\"adel and Meer, others (see, e.g., \cite{CuckerM99,HansenM06,Meer00}) have shed more light upon \emph{the descriptive complexity over the reals} mirroring the development of classical descriptive complexity.
 In addition to metafinite structures, the connection between logical definability encompassing numerical structures and computational complexity has received attention  in \emph{constraint databases} \cite{BENEDIKT2003169,GradelK99,Kreutzer00}. A constraint database models, e.g., geometric data by combining a numerical \emph{context structure}, such as the real arithmetic, with a finite set of quantifier-free formulae defining infinite database relations~\cite{KanellakisKR95}.

In this paper we investigate the descriptive complexity of so-called \emph{probabilistic independence logic} in terms of the BSS model of computation and the existential theory of the reals.
Probabilistic independence logic is a recent addition to the vast family of  new logics in \emph{team semantics}. 
In team semantics \cite{vaananen07}  formulae are evaluated with respect to sets of assignments which are called teams. During the past decade  research
on team semantics has flourished with interesting connections to fields such as database theory  \cite{HannulaK16}, statistics \cite{CoranderHKPV16},  hyperproperties  \cite{kmvz18}, and quantum information theory \cite{Hyttinen15b}, just to mention a few examples. The focus of this article is probabilistic team semantics that extends team based logics with probabilistic dependency notions. While the first ideas of probabilistic teams trace back to  \cite{galliani08,Hyttinen15b}, the systematic study of the topic was initiated by the works \cite{DurandHKMV18,HKMV18}. 

At the core of probabilistic independence logic
$\FO(\cpind)$ is the concept of conditional independence. The models of this logic are finite first-order structures but the notion of a team  is replaced by a probabilistic team, i.e.,  a discrete probability distribution over a  finite set of assignments. In \cite{HKMV18} it was observed that probabilistic independence logic is equivalent to a restriction of $\ESO_{\mathbb{R}}$ in which the weight functions are distributions. The exact complexity and relationship of $\FO({\cpind})$ to  $\ESO_{\mathbb{R}}$ and $\NPr$  was left as an open question; in this paper we present a (strict) sublogic of $\ESO_{\mathbb{R}}$ and a (strict) subclass of $\NPr$ that both capture $\FO({\cpind})$.

\textbf{\emph{Our contribution.}}
In this paper we introduce a novel variant of BSS machines
called Separate Branching BSS machines (S-BSS machines for short) and characterise its $\NP$
languages  (denoted by $\PNP{}{[0,1]}$) with 
$\peso{[0,1]}{+,\times,\leq, (r)_{r\in\RE}}$ that is a natural sublogic of
$\ESO_{\mathbb{R}}$. Likewise, we isolate a fragment $\exists [0,1]^\leq$ of the complexity class $\exists \RE$ and show that it coincides with the class of Boolean languages in $\PNP{0}{[0,1]}$. Moreover we establish a topological characterisation of the languages decidable by S-BSS machines; we show that, under certain natural restrictions, languages decidable by S-BSS machines are countable disjoint unions of closed sets in the usual topology of $\RE^n$. The topological characterisation separates the languages decidable by BSS machines and S-BSS machines, respectively. Moreover it enables us to separate the complexity classes $\PNP{0}{[0,1]}$ and $\NPrz$. Finally we show the equivalence of
the logics $\peso{[0,1]}{+,\times, \leq,\allowbreak 0,1}$  
and $\FO(\cpind)$, implying that 
$\FO(\cpind) \equiv \PNP{0}{[0,1]}$.
Table \ref{tbl:results} summarises the main results of the paper.

\textbf{\emph{Structure of the paper.}}
Section \ref{sec:preli} gives the basic definitions on descriptive complexity, BSS machines, and logics on $\RE$-structures required for this paper.
Section \ref{sec:characterisation} focuses in giving logical characterisations of variants of $\NP$ on S-BSS machines. In Section \ref{sec:SBSScharacterisation} we establish the aforementioned topological characterisation of S-BSS decidable languages. In Section \ref{sec:complexityhierarchy} we prove a hierarchy of the related complexity classes; in particular we separate $\PNP{0}{[0,1]}$ and $\NPrz$.
Section \ref{sec:teamsemantics} deals with probabilistic team semantics and establishes that $\FO(\cpind)\equiv\PNP{0}{[0,1]}$. Section \ref{sec:conclusion} concludes the paper.

\begin{table}
	\centering
	\begin{tabular}{ccccccc}
		& 			& $\bp{\PNP{0}{[0,1]}}$					& 		& $\bp{\NPrz}$						& &\\
$\NP$	& $\subseteq$ 	&  \rotatebox[origin=c]{-90}{$=$}${}^*$	&$\subseteq$ & \rotatebox[origin=c]{-90}{$=$} & $\subseteq$ & $\PSPACE$\\
		&			&  $\exists [0,1]^{\leq}$				& 		& $\exists\RE$						& & \\
		&&&&&&\\
	\end{tabular}
	\begin{tabular}{ccccccc}	
		&&$\PNP{0}{[0,1]}$& &$\NPrz$&& \\
		&& \rotatebox[origin=c]{-90}{$\equiv$}${}^*$ &$\subset^*$& \rotatebox[origin=c]{-90}{$\equiv$}&& \\
		&&$\peso{[0,1]}{+,\times, \leq,0,1}$& &$\ESOr[+,\times, \leq,0,1]$&&\\
		&& \rotatebox[origin=c]{-90}{$\equiv$}${}^*$ && && \\
		&& $\FO(\cpind)$ && && 
	\end{tabular}
\caption{Known complexity results and logical characterisations together with the main results
of this paper. The results of this paper are marked with an asterisk (*).  The top figure is with respect to Boolean inputs;
on the bottom figure, the inputs can include real numbers.}
\label{tbl:results}
\vspace{-4mm}
\end{table}

   \section{Preliminaries}\label{sec:preli}
   
     A vocabulary is \emph{relational} 
    (resp., \emph{functional})
    if it consists of only relation
    (resp., function)
    symbols. A structure is \emph{relational}
    if it is defined over a relational
    vocabulary. 
We let $\Varfo$ and $\Varso$ denote disjoint countable sets of first-order and function variables (with prescribed arities), respectively.
We write $\vec{x}$ to denote a tuple of first-order variables and $\vert \vec{x} \rvert$ to denote the length of that tuple.
The arities of function variables $f$ and relation symbols $R$ are denoted by $\ar(f)$ and $\ar(R)$, respectively.
If $f$ is a function with domain $\Dom(f)$ and $A$ a set, we define $f\upharpoonright A$ to be the function with domain   $\Dom(f)\cap A$ that agrees with $f$ for each element in its domain.
Given a finite set $D$, a function $f\colon D\to[0,1]$ that maps elements of $D$ to elements of the closed interval $[0,1]$ of real numbers such that $\sum_{s\in D}f(s)=1$ is called a \emph{(probability) distribution}.

 \subsection{$\mathbb{R}$-structures}

Let $\tau$ be a relational vocabulary.
A  $\tau$-structure is a  tuple $\A=(A, (R^\A)_{R\in\tau})$,
where $A$ is a  nonempty set and each $R^\A$ an $\ar(R)$-ary relation on $A$.
   The structure $\A$ is  \emph{a finite structure} if $\tau$ and $A$ are finite sets.
In this paper, we consider structures that enrich finite relational 
 $\tau$-structures by adding real numbers ($\RE$) as a second
 domain sort and functions that map tuples over  $A$ to $\RE$.
\begin{definition}
Let $\tau$ and $\sigma$ be respectively a  finite relational and a finite functional vocabulary, and let $X\subseteq \RE$. An $X$-structure of vocabulary $\tau\cup\sigma$ is a tuple 
\[
\A = (A,\RE, (R^\A)_{R\in\tau},
(g^\A)_{g\in\sigma}),
\]
where the reduct of $\A$ to $\tau$ is a finite relational structure, and 
 each $g^\A$ is a \emph{weight function} from $A^{\ar(g)}$ to $X$. Additionally, an $d[0,1]$-structure $\A$ is defined analogously, with the exception that the weight functions $g^\A$ are distributions.
\end{definition}
\emph{An assignment} is a total function $s:\Varfo \rightarrow A$ that assigns a value for each first-order variable. The modified assignment $s[a/x]$ is an assignment that maps $x$ to $a$ and agrees with $s$ for all other variables.

Next, we define a variant of functional existential second-order logic with numerical terms ($\ESO_\RE$) that is designed to describe properties of $\RE$-structures. As first-order terms we have only first-order variables. For a set   $\sigma$ of function symbols, the set of numerical $ \sigma$-terms $i$ is generated by the following grammar:
 \[
i ::= c\mid f(\vec{x}) \mid i\times i \mid i+i\mid
\SUM_{\vec{y}} \, i,
\]
where $c\in\RE$ is a real constant denoting itself, $f\in\sigma$, and $\vec{x}$ and $\vec{y}$ are tuples of first-order variables from $\Varfo$ such that the length of $\vec{x}$ is $\ar(f)$. The value of a numerical term $i$ in a structure $\A$ under an assignment $s$ is denoted by $[i]^\A_s$.
In addition to the natural semantics for the real constants, we have the following rules for the numerical terms:
\begin{align*}
&[f(\tuple x)]^\A_s:=f^{\A}(s(\tuple x)),     \quad   &&[i \times j]^\A_s := [i]^\A_s \cdot [j]^\A_s,\\
&[i+j]^\A := [i]^\A+[j]^\A, \hspace{2mm}
&&[\SUM_{\vec y}
\, i]^\A_s :=
\sum_{\vec a \in A^{|\vec y|}} [i]^\A_{s[\vec a/\vec y]},
\end{align*}
where $+,\cdot,\sum$ are the addition, multiplication, and summation of real numbers, respectively.

\begin{definition}[Syntax of $\ESO_\RE$]\label{esodef}
Let $\tau$ be a finite relational vocabulary and $\sigma$ a finite functional vocabulary. Let $O\sub \{+,\times, \SUM\}$,  $E \sub \{=,<,\leq\}$, and $C\sub \RE$.  The set of $\tau\cup\sigma$-formulae of $\esor{O,E,C}$ is defined via the grammar:
\begin{align*}
\phi ::= \ &   x=y \mid \neg x= y \mid i \mathrel e j \mid \neg
{i\mathrel e j} \mid R(\vec{x}) \mid \neg R(\vec{x}) \mid {} \\
 & \phi\land\phi \mid
  \phi\lor\phi \mid  \exists x\phi \mid \forall x \phi \mid  \exists f \psi,
 \end{align*}
where $i$ and $j$ are numerical $\sigma$-terms constructed using operations from $O$ and constants from $C$, and $e\in E$, $R\in\tau$ is a relation symbol, $f$ is a function variable,
 $x$ and $y$ are first-order variables and $\vec{x}$ a tuple of first-order variables,
  and $\psi$ is a $\tau\cup(\sigma\cup\{f\})$-formula of $\esor{O,E,C}$. 
\end{definition}
Note that the syntax of $\esor{O,E,C}$ allows first-order subformulae to appear only in negation normal form. This restriction however does not restrict the expressiveness of the language.

The semantics of $\esor{O,E,C}$ is defined via $\RE$-structures
and assignments analogous to first-order logic; note that
first-order variables are always assigned to a value in $A$
whereas functions map tuples over $A$ to $\RE$. In addition to the clauses of first-order logic, we have the following semantical clauses:
\begin{align}
&\A\models_s i\mathrel ej \Leftrightarrow [i]^\A_s \mathrel e
[j]^\A_s,  \quad \A\models_s \neg {i\mathrel e j} \Leftrightarrow
\A\not\models_s i\mathrel ej,\notag \\
&\A\models_s \exists f \phi \Leftrightarrow \A[h/f]\models_s \phi \text{ for some $h\colon A^{\ar(f)} \to \RE$,}\label{con2}
\end{align}
where $\A[h/f]$ is the expansion of $\A$ that interprets $f$ as $h$.

Given $S\sub \RE$, we define $\eso{S}{O,E,C}$ as the variant of $\esor{O,E,C}$ in which \eqref{con2} is modified such that $h\colon A^{\ar(f)} \to S$, and $\esod{O,E,C}$  as the variant in which \eqref{con2} is modified such that $h\colon A^{\ar(f)} \to [0,1]$ is a distribution, that is, $\Sigma_{\tuple a\in A^{\ar(f)}} h(\tuple a) =1$.
Note that in the setting of $\esod{O,E,C}$ the value $f^\A$ of a $0$-ary function symbol $f$ is always $1$.

\paragraph{Loose fragment.}
For both $S\subseteq \RE$ and $S= d[0,1]$,
define $\peso{S}{O,E,C}$ as the \emph{loose
fragment} of $\eso{S}{O,E,C}$ in which negated numerical atoms
$\neg {i\mathrel ej}$ are disallowed. We want to point out that as long as ${=}\in E$ and $0,1\in C$, the logic $\peso{S}{O,E,C}$ subsumes existential second-order logic over finite  structures (a precise formulation is given later by Proposition \ref{lem:esosub}).

\paragraph{Expressivity comparisons.}
Fix a relational vocabulary $\tau$ and a functional vocabulary $\sigma$. 
Let $\calL$ and $\calL'$ be some logics over $\tau\cup\sigma$ defined above, and let $X\subseteq \RE$ or $X=d[0,1]$.
For a formula $\phi\in\calL$, define $\struc^{X}(\phi)$ to be the class of $X$-structures $\A$ of vocabulary $\tau\cup\sigma$ such that $\A\models \phi$.
 We write $\calL \leq_X \calL'$ if for all sentences $\phi\in \calL$ there is a sentence $\psi\in \calL'$ such that $\struc^X(\phi)= \struc^X(\psi)$.
As usual, the shorthand $\equiv_X$
 stands for $\leq_X$ in both directions.
  For $X=\RE$, we write simply $\leq$ and $\equiv$. 
  
 In plain words, the subscript $S$ in $\eso{S}{O,E,C}$ constitutes the class of functions available for quantification, whereas the superscript $X$ in $\struc^X(\phi)$ constitutes the class of functions available for function symbols in the vocabulary. Thus,  $\phi \in \eso{S}{O,E,C}$ defines a class $\struc^{X}(\phi)$, even if $S$ and $X$ are different.
 \subsection{Blum-Shub-Smale Model}\label{sec:BSSmachines}

We will next give a definition of BSS machines (see e.g. \cite{BSSbook}). We define $\RE^* \dfn \bigcup \{\RE^n \mid n\in\N\}$. The \emph{size} $|x|$ 
of $x\in \RE^n$ is defined as $n$. The space $\RE^*$ can be seen as the real analogue of $\Sigma^*$ for a finite set $\Sigma$. We also define $\RE_*$ as the set of all sequences $x=(x_i)_{i\in \mathbb{Z}}$ where $x_i\in \RE$. The members of $\RE_*$ are thus of the form $(\ldots, x_{-2},x_{-1},x_0,x_1,x_2,\ldots )$. Given an element $x\in \RE^* \cup \RE_*$ we write $x_i$ for the $i$th coordinate of $x$. The space $\RE_*$ has natural shift operations. We define shift left $\sigma_l\colon \RE_* \to \RE_*$ and shift right $\sigma_r\colon\RE_* \to \RE_*$ as $\sigma_l(x)_i\dfn  x_{i+1}$ and $\sigma_r(x)_i\dfn x_{i-1}$.

\begin{definition}[BSS machines] \label{def:BSS} A BSS machine consists of an input space $\mathcal{I}=\RE^*$, a state space $\mathcal{S}=\RE_*$, and an output space $\mathcal{O}=\RE^*$, together with a connected directed graph whose nodes are labelled by $1, \ldots ,N$. The nodes are of five different types.
\begin{enumerate}
\item \emph{Input node}. The node labeled by $1$ is the only input node. The node is associated with a next node $\beta(1)$ and the input mapping $g_I: \mathcal{I} \to \mathcal{S}$.
\item \emph{Output node}. The node labeled by $N$ is the only output node. This node is not associated with any next node. Once this node is reached, the computation halts, and the result of the computation is placed on the output space by the output mapping $g_O: \mathcal{S}\to \mathcal{O}$.
\item \emph{Computation nodes.} A computation node $m$ is associated with a next node $\beta(m)$ and a mapping $g_m: \mathcal{S}\to \mathcal{S}$ such that for some $c\in \mathbb{R}$ and $i,j,k\in \mathbb{Z}$ the mapping $g_m$ is identity on coordinates $l\neq i$ and on coordinate $i$ one of the following holds:
\begin{itemize}
\item $g_m(x)_i =x_j+x_k$ (addition),
\item $g_m(x)_i = x_j - x_k$ (subtraction),
\item $g_m(x)_i = x_j\times x_k$ (multiplication),
\item $g_m(x)_i = c$ (constant assignment).
\end{itemize} 
\item \emph{Branch nodes.} A branch node $m$ is associated with nodes $\beta^-(m)$ and $\beta^+(m)$. Given $x\in \mathcal{S}$ the next node is $\beta^-(m)$ if $x_0\leq 0$, and $\beta^+(m)$ otherwise.
\item \emph{Shift nodes.} A shift node $m$ is associated either with shift left $\sigma_l$ or shift right $\sigma_r$, and a next node $\beta(m)$.
\end{enumerate}
The input mapping $g_I: \mathcal{I} \to \mathcal{S}$ places an
input $(x_1, \ldots ,x_n)$ in the state
\[(\ldots ,0,n,x_1, \ldots ,x_n, 0,\ldots )\in \mathcal{S},\]
 where the size of the input $n$ is located at the zeroth
 coordinate. The output mapping $g_O\colon \mathcal{S}\to
 \mathcal{O}$ maps a state to the string consisting of its first
 $l$ positive coordinates, where $l$ is the number of consecutive ones stored in the negative coordinates starting from the first negative coordinate.
 For instance, $g_O$ maps 
\[(\ldots ,2,1,1,1,n,x_1, x_2,x_3,x_4,\ldots )\in \mathcal{S},\]
 to $(x_1, x_2,x_3)\in \mathcal{O}$.
 A configuration at any moment of computation consists of a node
 $m\in \{1, \ldots ,N\}$ and a current state $x\in\mathcal{S}$.
The (sometimes partial) \emph{input-output} function $f_M:\RE^*\to \RE^*$ of a
machine $M$ is now defined in the obvious manner.
A function $f:\RE^*\to \RE^*$ is \emph{computable} if $f=f_M$ for some machine $M$. A language $L\sub \RE^*$ is \emph{decided} by a BSS machine $M$ if its characteristic function $\chi_L\colon \RE^*\to \RE^*$ is $f_M$.
 \end{definition}

\paragraph{Deterministic complexity classes.}
A  machine $M$ \emph{runs in (deterministic) time} $t\colon \N \rightarrow \N$,
if  $M$ reaches the output in $t(|x|)$
steps for each input $x\in \mathcal{I}$.
The  machine $M$ runs in \emph{polynomial time} if $t$ 
is a polynomial function.
 The complexity class
$\PTIME_\RE$ is defined as the set of all subsets of $\RE^*$ that
are decided by some  machine $M$ running in polynomial time.

\paragraph{Nondeterministic complexity classes.}
A language $L\sub \RE^*$ is \emph{decided nondeterministically}
 by a BSS machine $M$, if
 \[
 x \in L \quad\text{ if and only if }\quad f_M((x,x')) =1, \text{ for some $x'\in  \RE^*$},
 \]
when a slightly
different input mapping $g_I:\mathcal{I}\to \mathcal{S}$, which places
an input $(x_1, \ldots ,x_n,x'_1, \ldots ,x'_m)$ in the state
\[(\ldots ,0,n,m,x_1, \ldots ,x_n,x'_1, \ldots ,x'_m,\ldots )\in \mathcal{S},\]
where the sizes of $x$ and $x'$ are respectively placed on the first two coordinates,
is used. 
When we consider languages that a machine $M$ decides nondeterministically, we
call $M$ \emph{nondeterministic}. Sometimes when we wish to emphasize
that this is not the case, we call $M$ \emph{deterministic}.
Moreover, we say that $M$ is \emph{[0,1]-nondeterministic}, if
the guessed strings $x'$ are required to be from $[0,1]^*$.
L is \emph{decided in time} $t\colon \N \rightarrow \N$,
if, for every  $x \in L$, $M$ reaches the output $1$ in $t(|x|)$ steps for some $x'\in \RE^*$.
The machine \emph{runs in polynomial time} if $t$ is a polynomial function.
The class $\NP_\RE$ consists of those languages $L\subseteq \RE^*$
for which there exists a machine $M$ that nondeterministically decides $L$
in polynomial time.
Note that, in this case, the size of $x'$ above can be bounded by a polynomial (e.g., the running time of $M$) without altering the definition. The complexity class $\NP_\RE$ has many natural complete problems such as 4-FEAS, i.e., the problem of determining whether a polynomial of degree four has a real root \cite{blum1989}.

\paragraph{Complexity classes with Boolean restrictions.}
If we restrict attention to machines $M$ that may use only $c\in \{0,1\}$ in constant assignment nodes, then the corresponding complexity classes are denoted using an additional superscript $0$ (e.g., as in $\NP^0_\RE$).  Complexity classes over real computation can also be related to standard complexity classes. For a complexity class $\calC$ over the reals, the \emph{Boolean part} of $\calC$, written $\bp{\calC}$, is defined as $\{L\cap\{0,1\}^*\mid L\in \calC\}$.

\paragraph{Descriptive complexity.}
Similar to Turing machines, also BSS machines can be studied from the vantage point of descriptive complexity. To this end, finite $\RE$-structures are encoded as finite strings of reals using so-called rankings that stipulate an ordering on the finite domain. Let $\A$ be an $\RE$-structure over $\tau\cup\sigma$ where $\tau$ and $\sigma$ are relational and functional vocabularies, respectively. A \emph{ranking} of $\A$ is any bijection $\pi\colon\Dom(A)\to \{1, \ldots ,|A|\}$. A ranking $\pi$ and the lexicographic ordering on $\mathbb{N}^k$ induce a \emph{$k$-ranking} $\pi_k\colon\Dom(A)^k\to \{1, \ldots ,|A|^k\}$ for $k\in \mathbb{N}$.
Furthermore, $\pi$ induces the following encoding $\enc_\pi(\A)$. First we define $\enc_\pi(R^\A)$ and $\enc_\pi(f^\A)$ for $R\in \tau$ and $f\in \sigma$:
\begin{itemize}
\item Let $R\in \tau$ be a $k$-ary relation symbol. The encoding $\enc_\pi(R^\A)$ is a binary string of length $\lvert A \rvert^k$ such that the $j$th symbol in $\enc_\pi(R^\A)$ is $1$ if and only if $(a_1, \ldots ,a_k) \in R^\A$, where $\pi_k(a_1, \ldots ,a_k)=j$.
\item Let $f\in \sigma$ be a $k$-ary function symbol. The encoding $\enc_\pi(f^\A)$ is string of real numbers of length $\lvert A \rvert^k$  such that the $j$th symbol in $\enc_\pi(f^\A)$ is $f^\A(\vec{a})$, where $\pi_k(\vec{a})=j$.
\end{itemize}
The encoding $\enc_\pi(\A)$ is  then the concatenation of the string $(1,\ldots ,1)$ of length $|A|$ and the encodings of the interpretations of the relation and function symbols in $\tau\cup\sigma$.
We denote by $\enc(\A)$ any encoding $\enc_\pi(\A)$ of $\A$.

Let $\calC$ be a complexity class and $\eso{S}{O,E,C}$ a logic, where $O\sub \{+,\times, \SUM\}$,  $E \sub \{=,<,\leq\}$,  $C\sub \RE$, and $S\sub \RE$ or $S=d[0,1]$. Let $X\sub \RE$ or $X=d[0,1]$, and let $\calS$ be an arbitrary class of $X$-structures over $\tau\cup\sigma$ that is closed under isomorphisms.
We write $\enc(\calS)$ for the set of encodings of structures in $\calS$. Consider the following two conditions:
\begin{enumerate}[(i)]
\item $\enc(\calS)=\{\enc(\A)\mid \A \in \struc^X(\phi)\}$ for some $\phi\in \eso{S}{O,E,C}[\tau\cup\sigma]\}$,
\item $\enc(\calS)\in \calC$.
\end{enumerate}
If  $(i)$ implies $(ii)$, we write $\eso{S}{O,E,C} \leq_X \calC$, and if the vice versa holds, we write $\calC \leq_X \eso{S}{O,E,C} $. If both directions hold, then we write $\eso{S}{O,E,C} \equiv_X \calC$. We omit the subscript $X$ in the notation if $X=\RE$.

The following results due to Gr\"adel and Meer extend Fagin's theorem to the context of real computation.\footnote{Only the first equivalence is explicitly stated in \cite{GradelM95}. The second, however, is a simple corollary, using the fact that $0$ and $1$ can be identified in $\esor{+,\times,\leq}$; these two are the only idempotent reals for multiplication, and $0$ is the only idempotent real for addition.}
\begin{theorem}[\cite{GradelM95}]\label{thm:meer}
$\esor{+,\times,\leq,(r)_{r\in \RE}}\equiv \NP_\RE$ and \\
$\esor{+,\times,\leq}\equiv \NP^0_\RE$.
\end{theorem}

\subsection{Separate Branching BSS}
 We now define a restricted version of the BSS model which branches with respect to two separated intervals $(-\infty,\epsilon^-]$ and $[\epsilon^+,\infty)$. We will later relate these BSS machines to certain  fragments of $\ESO_{\RE}$ and the existential theory of the reals.

\begin{definition}[Separate Branching BSS Machine]
\emph{Separate branching BSS machines} (S-BSS machines for short) are otherwise identical to the BSS machines of Definition \ref{def:BSS}, except that the branch nodes are replaced with the following \emph{separate branch nodes}.
\begin{itemize}
\item \emph{Separate branch nodes.} A separate branch node $m$ is associated with $\epsilon_-,\epsilon_+ \in \RE$, $\epsilon_-<\epsilon_+$, and nodes $\beta^+(m)$ and $\beta^-(m)$. Given $x\in \mathcal{S}$ the next node is $\beta^+(m)$ if $x_0\geq \epsilon_+$,  $\beta^-(m)$ if $x_0\leq \epsilon_-$, and otherwise the input is rejected.
\end{itemize}
\end{definition}
Note that for a given S-BSS machine it is easy to write an equivalent BSS machine. A priori it is not clear whether the converse is possible; in fact, we will later show that in some cases the converse is not possible.

We can now define the variants of the complexity classes $\PTIMEr$, $\PTIMErz$, $\NPr$, and $\NPrz$ that are obtained by replacing BSS machines with S-BSS machines in the definitions of the complexity classes. Furthermore, we define $\NP_{[0,1]}$, and $\NP^0_{[0,1]}$ as the variants of $\NPr$, and $\NPrz$ in which the input $x$ may be any element from $ \RE^*$ but the guessed element $x'$ must be taken from $[0,1]^*$. Let $\calC$ be one of the aforementioned complexity classes. We define $\boole{\calC}{}{}$ to be the variant of $\calC$, where, instead of BSS machines, S-BSS machines are used. If $\calC$ includes the superscript $0$, this means that not only the parameter $c$ in constant assignment, but also $\epsilon_-$ and $\epsilon_+$ in separate branching are from $\{0,1\}$.

\section{Descriptive complexity of nondeterministic polynomial time in S-BSS}\label{sec:characterisation}
We now show that $\PNP{}{[0,1]}$ corresponds to a numerical variant of $\ESO$ in which quantified functions take values from the unit interval and numerical inequality atoms only appear positively. Later we show that both of these restrictions are necessary in the sense that removing either one lifts expressiveness to the level of $\esor{+,\times,\leq,(r)_{r\in \RE}}$ which captures $\NPr$. On the other hand, we give a logical proof, based on topological arguments, that $\PNP{}{[0,1]}< \NPr$. 
The proof of Theorem \ref{thm:main} is a nontrivial adaptation of
the proof of Theorem \ref{thm:meer} (see \cite[Theorem
4.2]{GradelM95}). In the proof we apply Lemma
\ref{lem:minus} and, by Proposition \ref{lem:esosub},
 assume without loss of generality built-in $\ESO$ definable
predicates on the finite part.

 Let $0$ and $1$ be distinct constants, $d$ a $(k+1)$-ary distribution, and $R$ a $k$-ary relation on a finite domain $A$ of size $n$. We say that $d$ is the \emph{characteristic distribution} of $R$ (w.r.t. $0$ and $1$) if $\tuple a\in R$ implies $d(\tuple a,1)=\frac{1}{n^k}$, and $\tuple a\notin R$ implies $d(\tuple a,0)=\frac{1}{n^k}$.
The next proposition implies that it is possible to simulate existential quantification of $\ESO$ definable predicates on the finite domain using function (or distribution) quantification; in particular, we may assume without loss of generality built-in predicates such as a linear ordering and its induced successor relation on the finite domain.
Clearly, any predicate that is $\ESO$-definable over finite structures is also $\ESO$-definable (w.r.t. the finite domain) over $\RE$-structures.

Below, we write $\mathrm{L}$-$\ESO_S[O,E,C,\exists X]$ to denote the extension of $\mathrm{L}$-$\ESO_S[O,E,C]$ by existential quantification of
relations over the finite domain with the usual semantics. 
\begin{proposition}\label{lem:esosub}
Let $\{0,1\}\sub S$ and $O,E,C$ be arbitrary.
For every formula  $\phi\in \mathrm{L}$-$\ESO_S[O,E,C,\exists X]$ there exist formulas $\phi' \in \mathrm{L}$-$\ESO_S[O, E\cup\{=\}, C\cup\{0,1\}]$ and
 $\phi''\in \mathrm{L}$-$\ESO_{d[0,1]}[O,E\cup\{=\},C]$ such that, for every $\RE$-structure $\mA$ and assignment $s$,
\[
\mA\models_s \phi \,\Leftrightarrow\, \mA\models_s \phi' \,\Leftrightarrow\, \mA\models_s \phi''.
\]

\end{proposition}
\begin{proof}
 The sentence $\phi'$ ($\phi''$, resp.) is obtained from $\phi$ by a translation that is the identity function, except that, for second-order variables $X$ of arity $k$, we rewrite the quantifications $\exists X$ as $\exists f_X$, where  $f_X$ is a $k$-ary (($k+1$)-ary, resp.) function variable, and the atoms 
 $X(\tuple x)$ and $\neg X(\tuple x)$ by $f_X(\tuple x)=1$ and $f_X(\tuple x)=0$ ($f_X(\tuple x,1)=u(\tuple x)$ and $f_X(\tuple x,0)=u(\tuple x)$, resp.), respectively.  Here, $u$ is the $k$-ary uniform distribution which is definable in  $\mathrm{L}$-$\ESO_{d[0,1]}{[=]}$ by $\forall \tuple x\tuple x' u(\tuple x)=u(\tuple x')$.
\end{proof}

\begin{lemma}\label{lem:minus}
If $\{0,1\}\sub C$, we have
$\peso{[0,1]}{+,\times,\leq,C}\equiv
\peso{[-1,1]}{+,\times,\leq,C}$.
\end{lemma}
\begin{proof}
Left-to-right direction is straightforward; the quantification
$\exists f\, \psi$ in
$\peso{[0,1]}{+,\times,\leq,C}$ can be simulated in
$\peso{[-1,1]}{+,\times,\leq,C}$ by the formula
\(
\exists f (\forall \tuple x \, 0 \leq  f(\tuple x)\wedge \psi).
\)

The converse direction is nontrivial. Let $\phi$ be an arbitrary
$\peso{[-1,1]}{+,\times,\leq,C}$-formula. We will show how to construct an equivalent
$\peso{[0,1]}{+,\times,\leq,C}$-formula $\phi'$. By the standard Skolemization argument we may assume that $\phi$ is in the prenex normal form. 
 Moreover, we assume that every atomic formula of the form $t_1\leq t_2$ is written such that $t_1$ and $t_2$ are multivariate polynomials where function terms $f(\vec{x})$ play the role of variables; this normal form is obtained by using the distributive laws of addition and multiplication. Let $M$ be the smallest set that includes every term of polynomials $t_1$ and $t_2$ such that $t_1\leq t_2$ is a subformula of $\phi$, and is closed under taking subterms.
 Clearly $M$ is a finite set, for its cardinality is bounded by the length of $\phi$. For each $p\in M$ with $m$ variables, we introduce an $m$-ary function $g_p$ that will be interpreted as the sign function for the term $p$. Let $\vec{x}_p$ be the related tuple of variables. The idea is that $g_p(\vec{a})=0$ ($g_p(\vec{a})=1$) if $p(\vec{a})<0$ ($p(\vec{a})\geq 0$).

We are now ready to define the translation $\phi\mapsto \phi'$, where
\[
\phi = \exists  f_1 \ldots \exists f_m Q_1 x_1 \ldots Q_n x_n\, \psi
\]
is in the normal form mentioned above. We define
\[
\phi' \dfn \bigexists_{p\in M} g_p \exists  f_1 \ldots \exists f_m Q_1 x_1 \ldots Q_n x_n (\theta\land \psi^\circ),
\] 
where the recursively defined translation ${}^\circ$ is homomorphic for the Boolean connectives and identity for first-order literals.

For atomic formulae $t_1\leq t_2$ of the form $s_1+\dots+ s_l\leq r_1+\dots +r_m$ the translation is defined as follows. The translation makes certain that every term (of polynomial) of the inequation after the translation has a non-negative value; this is done by moving terms to the other side of the inequation. Denote $\mathcal{I}=\{1,\dots,l\}$ and $\mathcal{J}=\{1,\dots,m\}$, and define $(t_1\leq t_2)^{\circ}$ as
\begin{align*}
\bigvee_{\substack{I\subseteq \mathcal{I} \\ J\subseteq \mathcal{J}}} &\Big( \bigwedge_{\substack{i\in I \\ j\in J}} g_{s_i}(\vec{x}_{s_i})=1 \land g_{r_j}(\vec{x}_{r_j})=1\\
 &\land \bigwedge_{\substack{i\in \mathcal{I}\setminus I \\ j\in \mathcal{J}\setminus J}}  g_{s_i}(\vec{x}_{s_i})=0 \land g_{r_j}(\vec{x}_{r_j})=0 \\
 &\land \sum_{i\in I} s_i + \sum_{j\in \mathcal{J}\setminus J} r_j \leq \sum_{i\in \mathcal{I}\setminus I} s_i + \sum_{j\in J} r_j \Big).
\end{align*}
Finally the subformula $\theta$ makes sure that the signs of the terms in $p\in M$ propagate correctly from subterms to terms. Define $\theta$ as
\begin{align*}
&\bigwedge_{\substack{p\in M \\ c\in M\cap [0,\infty]\\ d\in M\cap [-\infty,0)}} \forall \vec{x} \big( g_p(\vec{x})=0 \lor g_p(\vec{x})=1 \big) \land g_c=1 \land g_d=0 \\
&\land \bigwedge_{\substack{p,q,r\in M \\ p= q\times r}} \Big( \big( g_q(\vec{x}_q)=g_r(\vec{x}_r) \land g_p(\vec{x}_p)=1 \big)\\
&\quad\quad\lor \big( g_q(\vec{x}_q)=0 \land g_r(\vec{x}_r)=1 \land g_p(\vec{x}_p)=0 \big)\\
&\quad\quad\lor \big( g_q(\vec{x}_q)=1 \land g_r(\vec{x}_r)=0 \land g_p(\vec{x}_p)=0 \big)\Big).
\end{align*}
Note that the sign function maps terms of value $0$ to either $0$ or $1$, since for the purpose of the construction the sign of $0$ valued terms does not matter. 
\end{proof}

\begin{theorem}\label{thm:main}
$\peso{[0,1]}{+,\times,\leq,(r)_{r\in \RE}}\equiv \PNP{}{[0,1]}$.
\end{theorem}
\begin{proof}
\emph{\textbf{Right-to-left direction.}}
Suppose $L\in \PNP{}{[0,1]}$ is a class of $\RE$-structures that is closed under isomorphisms. By Lemma \ref{lem:minus} it suffices to construct an
$\peso{[-1,1]}{+,\times,\leq,\RE}$ sentence $\phi$  such that $\mA\models \phi$ iff $\mA\in L$ for all $\RE$-structures $\mA$.
Let $M$ be an S-BSS machine such that $M$ consists of $N$ nodes, and for each input $x$ it accepts $(x,x')$ for some $x'\in [0,1]^*$ in time $|x|^{k^*
 }$ iff $x=\enc(\mA)$ for some $\mA\in L$, where $k^*$ is some fixed natural number. We may assume that $|x'|$ is of size $|x|^{k^*}$. 
Let $k$ be a fixed natural number such that $\vert x\rvert^{k^*}\leq \vert A\rvert^k$; such a $k$ always exists since $\lvert\enc(\mA)\rvert$ is polynomial in $\lvert A\rvert$.
The computation of $M$ on a given input $\enc(\mA)$ can be represented using functions $f:A^{2k+1}\to (-1,1)$, $g:A^{2k+1}\to (0,1] $, and $h_1,\ldots ,h_N:A^{k}\to \{0,1\}$ such that
\begin{enumerate}[(a)]
\item $f(\tuple s,\tuple t)/g(\tuple s,\tuple t)$
is the content of register $\tuple s$ at time $\tuple t$; 
\item $h_i(\tuple t)$ is 1 if $i$ is the node label at time $\tuple t$,
and 0 otherwise.
\end{enumerate}
Note that $\tuple s$ is $(k+1)$-ary because we need to store $|A|^k$ positive and negative register contents. We may assume $k$ such that registers with index greater than $|A|^k$ do not contribute to the final outcome, i.e., their contents are never shifted to registers associated with the nodes of $M$.
Construct a formula 
 \[
 \psi(f,g,h):=\theta_{\textrm{pre}} \wedge \theta_{\textrm{initial}}\wedge \theta_{\textrm{comp}}\wedge \theta_{\textrm{accept}}
 \]
  of  $\peso{[-1,1]}{+,\times,\leq,(r)_{r\in \RE}}$  such that 
$\mA\models \exists fgh\, \psi$ iff $M$ accepts $\enc(\mA)$.
By Proposition \ref{lem:esosub} we may assume a built-in ordering $\fleq$ and its induced successor relation $S$ and constants $0,1,\max$ on the finite domain. 
Likewise, we may extend $\fleq$ to order also $k$-tuples from the finite domain. Under such ordering we then write $\tuple x+1$ ($\tuple x-1$) for the  element succeeding (preceding) a $k$-tuple $\tuple x$, and $\vec{n}$ for the $n$-th $k$-tuple. 
First, $\theta_{\textrm{pre}}$ is the conjunction of a formula stating that the ranges of $g$ and $h$ are as stated, and another formula
\begin{align}\label{eq:frac}
\forall \tuple y\, f(\tuple y)^2+g(\tuple y)=1,
\end{align}
where $f(\tuple y)^2$ is a shorthand for $f(\tuple y)\times f(\tuple y)$. Observe that \eqref{eq:frac} implies
\[\frac{f(\tuple y)}{g(\tuple y)} = \frac{f(\tuple y)}{(1- f(\tuple y)^2)}.\]  Also, $x \mapsto  x/(1-x^2)$ is a bijection from $(-1,1)$ to $\mathbb{R}$.
 That the range of $f$ is $(-1,1)$ will follow from the remaining conjuncts of $\psi$, described below.

\noindent
\textit{Initial configuration.}  We give a description of $\theta_{\textrm{initial}}$ such that 
\begin{multline}\label{eq:zero}
(\mA,f,g,\tuple h)\models \theta_{\textrm{initial}}\\
\text{iff $(f,g,\tuple h)$ satisfies (a) \& (b) at time $\tuple 0$}.
\end{multline}
For clause (b) it suffices to add  to $\theta_{\textrm{initial}}$
\[h_1(\tuple 0)=1\wedge h_2(\tuple 0)=0\wedge \ldots \wedge h_N(\tuple 0)=0.\] 
Consider then clause (a). We denote by $\tuple s_0$ the $\lfloor |A^{k+1}|/2\rfloor$th $k+1$-tuple with respect to $\fleq$. 
  The sequence $\tuple s_0$, which is clearly definable in $\ESO$, now represents the zeroth coordinate of $R_*$.
  To encode that $|x|$ is placed on zeroth coordinate 
 we add to $\theta_{\textrm{initial}}$ 
\begin{align}\label{eq:count}
\exists \epsilon\exists  f_{\rm count} \Big(&  f_{\rm count}(0)= \epsilon\\\nonumber
&\land \forall xy \big(S(x,y)\to f_{\rm count}(y)=f_{\rm count}(x)+ \epsilon\big) \\\nonumber
&\land  f_{\rm count}(\max)=1 
 \land f(\tuple s_0,\tuple 0)=p(1/\epsilon)\times g(\tuple s_0,\tuple 0)\Big),\nonumber
 \end{align}
  where $\epsilon$ is a nullary function variable (i.e., a real from $[-1,1]$),
  $p$ is a polynomial such that $|\enc(\mA)|=p(|A|)$, and the last conjunct of \eqref{eq:count} is a shorthand for
  \[
  \epsilon^{\deg(p)}\times  f(\tuple s_0,\tuple 0) = p^*(\epsilon)\times g(\tuple s_0,\tuple 0),
  \]
  where $\deg(p)$ is the degree of the polynomial $p$, and $p^*$ is the polynomial obtained by multiplying $p$ by $\epsilon^{\deg(p)}$ (that is $\epsilon^{\deg(p)}\times p(1/\epsilon) = p^*(\epsilon)$). 
  It follows from \eqref{eq:frac} and \eqref{eq:count} that $f(\tuple s_0,\tuple 0)\in (-1,1)$ and $f(\tuple s_0,\tuple 0)/g(\tuple s_0,\tuple 0)=|\enc(\mA)|$. To encode that $|x'|$ is placed on the first coordinate we also add to $\theta_{\textrm{initial}}$ a formula stipulating that  $f(\tuple s_0,\tuple 0)^{k^*}/g(\tuple s_0,\tuple 0)^{k^*}=f(\tuple s_0+1,\tuple 0)/g(\tuple s_0+1,\tuple 0)$.

   Let $f^*\in \tau$ be a function symbol and let $r_{f^*}$ be a natural number that indicates the starting position of the encoding of $f^*$ in $\enc(\mA)$. Clearly $r_{f^*}$ is a definable real number as it is the value of a fixed univariate polynomial. We use the shorthand $\vec{s} = \vec{y} + r_{f^*}$ to denote that in the ordering of $k$-tuples (induced from $\fleq$) the ordinal number of $\vec{s}$ is the sum of the ordinal number of $\vec{y}$ and $r_{f^*}$. Clearly $\vec{s} = \vec{y} + r_{f^*}$ is expressible in our logic. We then add the following to $\theta_{\textrm{initial}}$: 
   \begin{equation}\label{equ}
   \forall\vec{s}\vec{y} \bigwedge_{f^* \in \tau} \Big( \vec{s} = \vec{y} + r_{f^*} \rightarrow  \big(f(\tuple s,\tuple 0)=f^*(\tuple y)\times g(\tuple s,\tuple 0)\big) \Big)
   \end{equation}

Note that \eqref{eq:frac} and \eqref{equ} imply that $f(\tuple s,\tuple 0)\in(-1,1)$; for, by \eqref{eq:frac}, $|f(s,0)|=1$ leads to $g(s,0)=0 $ which contradicts \eqref{equ}.
The interpretations of relations in $\sigma$ are treated analogously. For all the remaining positions $\tuple s>\tuple s_0$ we stipulate that $0\leq f(\tuple s,\tuple 0)\leq g(\tuple s,\tuple 0)$, and 
 for all positions $\tuple s<\tuple s_0$ we stipulate that $f(\tuple s,\tuple 0)=0$.  In the first case $f(\tuple s,\tuple 0)/g(\tuple s,\tuple 0)$ is some value guessed from the unit interval $[0,1]$ and in the second case it is $0$. We conclude that \eqref{eq:zero} holds by this construction.

\noindent
\textit{Computation configurations.}
Then we define  $\theta_{\textrm{comp}}$ such that
\begin{multline}\label{eq:greater}
(\mA,f,g,\tuple h)\models \theta_{\textrm{comp}}\\
\text{ iff $(f,g,\tuple h)$ satisfies (a) and (b) at time $\tuple t>\tuple 0$}.
\end{multline}
We let 
\begin{align*}
\theta_{\textrm{comp}} \dfn  
\forall \tuple s \, \tuple t \Big(&\bigvee_{1 \leq m< m'\leq N} \big(h_m(\tuple t)=0 \vee h_{m'}(\tuple t)=0 \big)\wedge\\
& \bigvee_{1\leq m\leq N} \big(h_m(\tuple t)=1\wedge \theta_m\big)\Big),
\end{align*}
 where each $\theta_m$ describes the instruction of node $m$. Suppose $m$ is a computation node associated with a mapping $g_m$ that is the identity on coordinates $l\neq i$ and on coordinate $i$ defined as $g_m(x)_i=x_j+x_k$. Let us write $f_{\tuple s,\tuple t}$ and $g_{\tuple s,\tuple t}$ for $f(\tuple s,\tuple t)$ and $g(\tuple s,\tuple t)$, and  $\tuple s_i,\tuple s_j,\tuple s_k$ for the tuples that correspond to the $i$th, $j$th, and $k$th input coordinates. Clearly, these tuples are definable. We define
 \begin{align*}
 \theta_m \dfn &\,h_{\beta(m)}(\tuple t+1)=1 \wedge f_{\tuple s_i,\tuple t+ 1}\times g_{\tuple s_j,\tuple t}\times g_{\tuple s_k,\tuple t}\\
 & =  g_{\tuple s_i,\tuple t+1}\times (  f_{\tuple s_j,\tuple t}\times g_{\tuple s_k,\tuple t} + g_{\tuple s_j,\tuple t} \times f_{\tuple s_k,\tuple t} ) \wedge \\
 & \tuple s\neq \tuple s_i \to (f_{\tuple s,\tuple t+1} = f_{\tuple s,\tuple t} \wedge g_{\tuple s,\tuple t+1} = g_{\tuple s,\tuple t} ).
  \end{align*} 
 The other computation nodes are described analogously. 
 For a shift left node $m$ we define 
 \begin{align*}
 \theta_m \dfn &\,h_{\beta(m)}(\tuple t+1)=1 \,\wedge\\
 & \tuple s < \tuple \max \to (f_{\tuple s,\tuple t+1} = f_{\tuple s+1,\tuple t} \wedge g_{\tuple s,\tuple t+1} = g_{\tuple s+1,\tuple t} ),
 \end{align*} 
 and the case for shift right node is analogous.
 For  a separate branch node $m$ we define
\begin{align*}
 \theta_m \dfn &\,\Big(\big(h_{\beta^+(m)}(\tuple t+1)=1 \wedge f_{\tuple s_0,\tuple t}\geq \epsilon^+\big) \vee\\
  &\big(h_{\beta^-(m)}(\tuple t+1)=1 \wedge f_{\tuple s_0,\tuple t}\leq \epsilon^-\big)\Big)\wedge \\
   &  f_{\tuple s,\tuple t+1} = f_{\tuple s,\tuple t} \wedge g_{\tuple s,\tuple t+1} = g_{\tuple s,\tuple t}.
 \end{align*} 
 Our formulae now imply that \eqref{eq:greater} follows by the construction. In particular, keeping the values of $f$ in $(-1,1)$ ensures that the arithmetical operations are encoded correctly.

Finally, to express that the value of the characteristic function $f_M$ is $1$ we may stipulate without loss of generality that coordinates $-2,-1,1$ respectively contain $0,1,1$; we also need to state that the machine is in node $N$ at the last step:
\begin{align*}
\theta_{\rm accept}:= & h_N(\tuple \max) =1 \land f_{\tuple s_0+1,\tuple \max}= g_{\tuple s_0+1,\tuple \max}\\
& \wedge f_{\tuple s_0-1,\tuple \max}= g_{\tuple s_0-1,\tuple \max}\wedge f_{\tuple s_0-2,\tuple \max}= 0.
 \end{align*}
We conclude that  $\mA \models \exists fg\tuple h\, \psi$
iff $M$ accepts $\enc(\mA)$.

\textbf{\emph{Left-to-right direction.}}
Let $\phi\in\peso{[0,1]}{+,\times,\leq,\RE}$  be a sentence over some vocabulary $\sigma\cup\tau$. As in the previous lemma, we may assume that $\phi$ is of the form 
\[
\exists  f_1 \ldots \exists f_m Q_1 x_1 \ldots Q_n x_n\,\psi,
\] where
 $\psi$ is quantifier-free.
 We may further may transform $\phi$ to an equivalent form 
\begin{equation}\label{eq:skolem2}
\exists  f_1 \ldots \exists f_m \exists g_{i_{l+1}} \ldots
\exists g_{i_n} \forall x_{i_{1}} \ldots \forall x_{i_l} \,\psi',
\end{equation}
 where $g_{i_j}$ are Skolem functions on the finite domain and $\psi'$ is obtained from $\psi$ by replacing each occurrence of $x_{i_j}$, $l+1\leq j\leq n$, with $g_{i_j}(\tuple x_j)$.
Note that \eqref{eq:skolem2} is an intermediate expression which is not anymore in $\peso{[0,1]}{+,\times,\leq,\RE}$.
  We may assume $\psi'$ is in disjunctive normal form $\bigvee_{i\in I} C_i$, where $I$ is a finite set of indices.

  Suppose the relational and function symbols in $\sigma\cup \tau\cup \{f_1, \ldots ,f_m\}$ are of arity at most $n'\geq n$. First, a fixed initial segment of negative coordinates is allocated with the following intention:
\begin{itemize}
\item one coordinate $a$ for separate branching,
\item three coordinates $i,j,k$ for numerical identity atoms,
\item two sequences of coordinates $\tuple b=(b_1, \ldots ,b_n)$ and $\tuple c=(c_1,\ldots ,c_{n'})$ for elements of the finite domain.
\end{itemize}
 We construct a machine $M$ which runs in polynomial time and accepts  $(x,x')$ iff 
\begin{enumerate}
\item $x=\enc(\mA)$ where $\mA$ is a model over $\sigma\cup\tau$, and
\item $(x,x')$ is a concatenation of $\enc((\mA,\tuple f,\tuple g))$
 and indices $i_{\tuple a}\in I$ such that $(\mA,\tuple f,\tuple g,\tuple a)\models C_{i_{\tuple a}}$ for each $\tuple a\in A^{l}$.
\end{enumerate}
We may suppose that $\tuple f$ and $(\tuple g,(i_{\tuple a})_{\tuple a\in A^l})$ are respectively encoded as strings of reals and integers.

 Let $p'$ be a polynomial such that for each $\mA$ over $\sigma\cup \tau$ we have $p'(|A|)=\enc(\mA)$.
The machine first checks whether there is a natural number $d$ such that $p'(d)=|x|$. For this, it first sets $x_i\gets 1$ and $x_a\gets x_0- p'(x_i)$, where initially $x_0=|x|$. If $x_a=0$, then $x_0\gets x_i$, and if $x_a\geq 1$, then $x_i \gets x_i+1$ and the process is repeated. Otherwise, if $x_a\notin\{0\}\cup[1,\infty)$, the input is rejected. This type of branching can be implemented repeating separate branching twice.
 Provided that the input is not rejected, this process terminates with $x_0=d$ where  $p'(d)=|x|$. The machine then checks whether item 1 holds; given $\lvert \mA \rvert$ this is straightforward.
 Checking that $(x,x')$ is a concatenation of $\enc((\mA,\tuple f,\tuple g))$, for some functions $\tuple f,\tuple g$, and some indices $i_{\tuple a}$ is analogous.

 It remains to be checked that the last claim of item 2 holds. We go through all tuples $\vec{a}\in A^l$, calculate the values of the Skolem functions, and check that the disjunct $C_{i_{\tuple a}}$ holds for the calculated value of the variables.
 For each $\vec{a} = (a_1, \ldots ,a_l) \in \{0,\ldots ,d-1\}^l$, placed on the coordinates $b_1, \ldots ,b_l$, the machine uses $x_0$ and $\tuple c$ for retrieving and placing $g_{i_{l+1}}(\tuple a_{l+1}), \ldots ,g_{i_{n}}(\tuple a_{n})$
 on the coordinates $b_{l+1}, \ldots ,b_n$.
 The machine then retrieves the index $i_{\tuple a}$ and checks whether $C_{i_{\tuple a}}$ holds true with respect to the values on coordinates $\tuple b$. Once this process is completed for all value combinations $(a_1, \ldots ,a_l) \in \{0,\ldots ,d-1\}^l$ the computation halts with accept.

The contents of the input are accessed using shifts which fix the contents of the allocated coordinates. That is, we use operations $\sigma_l^{X}$, where $X$ is a finite set of coordinates, such that $\sigma_l^X(x)_i=x_i$ if $i\in X$, and otherwise $\sigma_l^X(x)_i=x_j$ where $j=\min\{k\in \mathbb{N}\mid  k> i,k\notin X\}$. For instance, $\sigma_l^{\{0\}}$
is obtained by first swapping $x_0$ and $x_1$ and then shifting left. 

 Also, if $C_{i_{\tuple a}}$ contains a numerical atom $f(\tuple t_0)\leq g(\tuple t_1)\times h(\tuple t_2)$, then the values of its constituent function terms with respect to $\tuple b$ are placed on coordinates $i,j,k$. The machine then sets $x_a\gets x_i - x_j\times x_k$, and if $x_a \leq 0$, then it continues to the next atom in $C_{i_{\tuple a}}$, and else it rejects.
If $C_{i_{\tuple a}}$ contains a relational atom $R(\tuple x_0)$, then the value of its characteristic function with respect to $\tuple b$ is placed on coordinate $a$. If $x_a=1$, then the machine moves to the next atom in $C_{i_{\tuple a}}$, and else it rejects. Negated relational atoms are treated analogously, and the stated branching is straightforward to implement with separate branch nodes.

It follows from our construction that $M$ runs in polynomial time and accepts $(x,x')$ iff items 1 and 2 hold.
 Hence, we conclude that $\peso{[0,1]}{+,\times,\leq,(r)_{r\in \RE}}\leq \PNP{}{[0,1]}$.
\end{proof}
Suppose we above consider (i) guesses from $\RE$ instead of $[0,1]$, or (ii) BSS instead of S-BSS machines. Then slightly modified proofs yield (i)  $\peso{\RE}{+,\times,\leq,(r)_{r\in \RE}}\equiv \PNP{}{\RE}$, and (ii) $\eso{[0,1]}{+,\times,\leq,(r)_{r\in \RE}}\equiv \NP_{[0,1]}$. Furthermore, logical constants $r\in \RE\setminus \{0,1\}$ are only needed to capture $c$ in constant assignment and $\epsilon^+,\epsilon^-$ in separate branching, and for the converse direction only those machine constants $r\in \RE\setminus \{0,1\}$ which explicitly occur in the logical expression are needed. Thus we obtain the following corollary.  
\begin{corollary}\label{cor:characterisation}
~
\begin{enumerate}
\item $\peso{\RE}{+,\times,\leq,(r)_{r\in \RE}}\equiv\PNP{}{\RE}$,
\item $\peso{\RE}{+,\times,\leq,0,1}\equiv\PNP{0}{\RE}$,
\item $\peso{[0,1]}{+,\times,\leq,0,1}\equiv\PNP{0}{[0,1]}$\label{this},
\item $\eso{[0,1]}{+,\times,\leq,(r)_{r\in \RE}}\equiv\NP_{[0,1]}$,
\item $\eso{[0,1]}{+,\times,\leq,0,1}\equiv\NP_{[0,1]}^{0}$.
\end{enumerate}
\end{corollary}
In the following two sections we investigate how S-BSS computability relates to BSS computability, and in particular how $\PNP{}{[0,1]}$ relates to $\NPr$. On the one hand it turns out that  $\PNP{}{[0,1]}$ is strictly weaker than $\NPr$. On the other hand both obvious strengthenings of $\PNP{}{[0,1]}$, namely  $\PNP{}{\RE}$ and $\NP_{[0,1]}$, collapse to $\NPr$.

\section{Characterisation of S-BSS decidable languages}\label{sec:SBSScharacterisation}
We give a characterisation of languages decidable by S-BSS machines using the ideas from the previous section.
The goal of this section is to establish the following theorem:
\begin{theorem}\label{thm:SBSScharacterisation}
Every language that can be decided by a) a deterministic S-BSS machine, or b) a $[0,1]$-nondeterministic S-BSS machine in time $t$, for some function $t\colon \N\rightarrow \N$, is a countable disjoint union of closed sets in the usual topology of $\RE^n$.
\end{theorem}
The result complements an analogous characterisation of BSS-decidable languages thus giving insight on the difference of the computational powers of BSS machines and S-BSS machines.
\begin{theorem}[{\cite[Theorem 1]{BSSbook}}]
	Every language decidable by a (deterministic) BSS machine is a countable disjoint union of semi-algebraic sets. 
\end{theorem}
These characterisations are based on the fact that  the computation of BSS and S-BSS machines can be encoded by formulae of first-order real arithmetic.

\paragraph{Existential theory of the real arithmetic.}
Formulae of the \emph{existential real arithmetic}
are given by the grammar
\begin{equation}\label{eq:synt}
\phi::=i\leq i \mid i<i \mid \phi\wedge \phi \mid \phi \vee \phi \mid \exists x\phi,
\end{equation}
where $i$ stands for numerical terms given by the grammar
\[
i::= 0 \mid 1\mid x \mid i\times i\mid i+i,
\]
where $x$ is a first-order variable.
 The semantics is defined over a fixed structure $(\RE, +, \times, \leq, 0, 1)$ of real arithmetic in the usual way. Relations definable by such formulae with additional real constants are called semi-algebraic.

Let $M$ be an S-BSS machine and $n,t\in \N$ positive natural numbers. We denote by $L^n_t(M)$  ($L^n_{\leq t}(M)$, resp.) the set of strings $s\in \RE^n$ accepted by $M$ in time exactly (at most, resp.) $t$, and
define $L^n(M) \dfn  L(M)\cap \RE^n$. 
The following restricted fragment of $\exists \FO$ is enough to  encode S-BSS computations.

\paragraph{Existential theory of the loose $[0,1]$-guarded real arithmetic.}
Formulae of the \emph{existential loose $[0,1]$-guarded real arithmetic}
are defined as in \eqref{eq:synt}, but without $i<i$ and replacing $\exists x\phi$ with $\exists x (0\leq x\leq 1\wedge \phi)$.

\begin{lemma}\label{new}
Given a  deterministic or $[0,1]$-nondeterminis-tic S-BSS machine  $M$ and  positive $n,t\in \N$ it is possible to construct, in polynomial time, formulas $\phi$ and $\psi$ of loose $[0,1]$-guarded real arithmetic, with free variables $x_1,\dots, x_n$, that may use real constants used in $M$ such that 
\begin{align*}
	\{\big(s(x_1),\dots, s(x_n)\big) \mid (\RE, +, \times, \leq, (r)_{r\in\RE}) &\models_s \phi\} = L^n_t(M),\\
	\{\big(s(x_1),\dots, s(x_n)\big) \mid (\RE, +, \times, \leq, (r)_{r\in\RE}) &\models_s \psi\} = L^n_{\leq t}(M).
\end{align*}
\end{lemma}
\begin{proof}
For a given input  of length $n$, the computation of $M$  consists of $t$ many configurations $\vec c_1,\dots \vec c_{t}$ of $M$, where $\vec c_1$ and  $\vec c_{t}$ are the initial configuration and a terminal configuration, respectively, and, for $1\leq m < t$,  $\vec c_{m+1}$ is a successor configuration of $\vec c_{m}$. Each configuration is a string of real  numbers of length $\mathcal{O}(t)$. We can use a similar technique as in the \emph{right-to-left} direction of Theorem \ref{thm:main} and encode the contents of registers by pairs of real numbers from the unit interval $[0,1]$. In order to encode the computation, it suffices to encode the values of $\mathcal{O}(t^2)$ registers; thus $\mathcal{O}(t^2)$ variables suffice.
	We then construct a formula of existential loose $[0,1]$-guarded real arithmetic of size $\mathcal{O}(t^2)$ that first existentially quantifies $\mathcal{O}(t^2)$-many variables in order to guess the whole computation of $M$ on the given input and then expresses, using perhaps at most polynomially many extra variables, that the computation is correct and accepting.  We omit further details, for the encoding is done in a similar manner as in the \emph{right-to-left} direction of Theorem \ref{thm:main}.
\end{proof}

Given a deterministic S-BSS machine $M$, it is easy to see that the sets  $L^n_t(M)$, for $n,t\in \N$, are disjoint. However, the same does not need to hold for nondeterministic machines, for the time it takes to accept an input string $x$ might depend on the guessed value for the string $x'$ (and there may be multiple accepting runs with different values for $x'$).
This problem can be evaded for languages $L$ that can be decided by a $[0,1]$-nondeterministic S-BSS machine $N$ in time $f$, for some function $f\colon \N\rightarrow \N$. In this case $L^n(N)= L^n_{\leq f(n)}(N)$, for each $n\in \N$.
Now since $L(M)=\bigcup_{n,t\in \N} L^n_t(M)$ and $L(N) =\bigcup_{n\in \N} L^n(N)$ where the unions are disjoint, we obtain the following characterisation.
\begin{theorem}\label{thm:SBSSER}
	Every language decidable by a) a deterministic S-BSS machine or b) a $[0,1]$-nondeterministic S-BSS machine in time $t$, for some $t\colon\N\rightarrow \N$, is a countable disjoint union of relations 
	defined by existential loose $[0,1]$-guarded real arithmetic formulae that may use real constants from some finite~set.
\end{theorem}

The rest of this section is dedicated on proving the following theorem, which together with Theorem \ref{thm:SBSSER} implies Theorem \ref{thm:SBSScharacterisation}.
\begin{theorem}\label{thm:PEFO}
	Every relation defined by some existential loose $[0,1]$-guarded real arithmetic formula $\phi(x_1,...,x_n)$ with real constants is closed in 
	$\RE^n$.
\end{theorem}

\paragraph{Point-set topology.}
The proof of the theorem relies on some rudimentary notions and knowledge from
point-set topology summarised
in the following two lemmas (for basics of point-set topology see, e.g., the monograph \cite{willard2004general}).
In order to simplify the notation, for a topological space $X$, we use $X$ to denote also the underlying set of the space. Likewise, in this section, we let [0, 1] denote the topological space that has domain [0, 1] and the metric of Euclidean distance.

\begin{lemma}\label{prop:topology}
	Let $X$ and $Y$ be topological spaces, $f\colon X\rightarrow Y$ a continuous function, $A$ and $B$ closed sets in $X$, and $C$ a closed set in $Y$. Then
	\begin{itemize}
		\item $X$, $A\cap B$, $A\cup B$, and $f^{-1}[C]$ are closed in $X$,
		\item the product $A\times C$ is closed in the product space $X\times Y$,
		\item if $Y\supseteq A$ is a subspace of $X$ then $A$ is closed in $Y$.
	\end{itemize}
\end{lemma}

\begin{lemma}\label{thm:projectionclosed}
	Let $X$ be a topological space, $Y$ a compact topological space, $A$ a closed set in the product space $X\times Y$, and $f$ the projection function $X\times Y \rightarrow X$. Then the image $f[A]$ of $A$ is closed in $X$.
\end{lemma}

\begin{proof}[Proof of Theorem \ref{thm:PEFO}]
	We prove the following claim by induction on the structure of the formulae: Let $\vec{x}$ be a $k$-tuple of distinct variables and $\phi(\vec{x})$ an existential loose $[0,1]$-guarded real arithmetic formula with real constants,
	and its
	free
	variables in $\vec{x}$. 
	The relation defined by $\phi(\vec{x})$ is closed in $\RE^k$.
	
	\begin{itemize}
		\item Assume $\phi=t_1 \leq t_2$.
		Recall that $t_1(\vec{x})$ and $t_2(\vec{x})$ are multivariate polynomials.
		Define $g(\vec{x})$ as the multivariate polynomial $t_1(\vec{x})-t_2(\vec{x})$ and consider the preimage $g^{-1}[(-\infty,0]]$. Since $(-\infty,0]$ is closed in $\RE$ and $g\colon \RE^k\rightarrow \RE$ is a continuous function, it follows that $g^{-1}[(\infty,0]]$ is closed.
		Clearly $g^{-1}[(-\infty,0]]$ is the relation defined by $\phi(\vec{x})$.
		\item The cases of disjunctions and conjunctions are clear, for the union and intersection of closed sets is closed.
		\item Assume $\phi=\exists y (0\leq y\leq 1\wedge \psi(\tuple x,y))$.
		Let  $R_\psi$ be the relation defined by $\psi(\vec{x},y)$, which by induction hypothesis is closed in $\RE^{k+1}$. Define $R'_\psi \dfn R_\psi \cap (\RE^k\times [0,1])$. Since $[0,1]$ is closed in $\RE$, it follows from
		Lemma~\ref{prop:topology} that $R'_\psi$ is closed both in $\RE^{k+1}$ and $\RE^k\times [0,1]$. Let $R^*_\psi$ be the projection of $R'_\psi$ to its $k$ first columns. Since $R'_\psi$ is closed in $\RE^{k}\times [0,1]$, and $[0,1]$ is a compact topological space, it follows from
		Lemma~\ref{thm:projectionclosed} that $R^*_\psi$ is closed in
		$\RE^k$. Clearly $R^*_\psi$ is the relation defined by $\psi(\vec{x})$.\qedhere
	\end{itemize}
\end{proof}

\section{Hierarchy of the complexity classes}\label{sec:complexityhierarchy}
The main result of this section is the separation of the complexity classes $\PNP{}{[0,1]}$ and $\NPr$. We have already done most of the work required for the separation as the result follows directly from the topological argument of Section \ref{thm:PEFO} that more generally separates S-BSS computations from BSS computations. 
The characterisations of Section \ref{sec:characterisation} then yield the separation of the related logics on $\RE$-structures.
We also give logical proofs implying that the obvious strengthenings of $\PNP{}{[0,1]}$ coincide with $\NPr$.
Finally we study the restriction of $\PNP{0}{[0,1]}$ on Boolean inputs and establish that it coincides with a natural fragment of $\exists\RE$.

\subsection{Separation of $\PNP{}{[0,1]}$ and $\NPr$}\label{sec:separations}
We can now use Theorem \ref{thm:PEFO} to prove the following:

\begin{theorem} The following separations hold:
\begin{enumerate}
\item $\PNP{0}{[0,1]}<\NPrz$ and  $\PNP{}{[0,1]}<\NPr$,
\item $\peso{[0,1]}{+,\times,\leq,0,1} < \esor{+,\times,\leq,0,1}$,
\item $\peso{[0,1]}{+,\times,\leq, (r_{r\in\RE})} < \esor{+,\times,\leq, (r)_{r_\in\RE}}$.
\end{enumerate}
\end{theorem}

\begin{proof}
	We prove 1. by showing that there are languages in $\NPrz$ that are not in $\PNP{}{[0,1]}$. The claims 2. and 3. then follow from the logical characterisations of Corollary \ref{cor:characterisation}.
	
	Let $L$ be a language in $\PNP{}{[0,1]}$ and $M$ an $\PNP{}{[0,1]}$ S-BSS machine such that $L(M)=L$.
	Let $p$ be a polynomial function that bounds the running time of $M$.
	Fix $n\in\N$. Now $L^n=L^n_{\leq p(n)}$.
	By Lemma \ref{new} $L^n_{\leq p(n)}$, and hence  $L^n$, is
	definable by an existential loose $[0,1]$-guarded real
	arithmetic formula $\phi(x_1,...,x_n)$ that uses real
	constants from $M$. By Theorem \ref{thm:PEFO} $L^n$ is a
	closed set in the product space $\RE^n$, which
	 is not true for all languages in $\NPrz$; for instance, a language $P$ consisting of all finite strings of positive reals can be decided in $\NPrz$ (using branching), but $P^n$ is not closed in $\RE^n$.
\end{proof}

\subsection{Robustness of $\NPr$}

We have just seen that $\PNP{}{[0,1]}$ is a
 complexity class strictly below $\NPr$.
We now give purely logical proofs implying that the obvious strengthenings of $\PNP{}{[0,1]}$ collapse to $\NPr$. The proofs are based on the logical characterisations established in Corollary \ref{cor:characterisation}.

The first obvious question is: Are $\PNP{}{\RE}$ and $\PNP{0}{\RE}$ strictly below $\NPr$ and $\NPrz$?
In logical terms this boils down to the expressivity of the logic $\peso{\RE}{+,\times,\leq,(r)_{r\in \RE}}$.
We answer to this question in the negative.

\begin{proposition}\label{prop:firstcollapse}
	$\peso{\RE}{+,\times,\leq,0,1} \equiv \eso{\RE}{+,\times,\leq}$ and 
	$\peso{\RE}{+,\times,\leq, (r)_{r\in  \RE}} \equiv \eso{\RE}{+,\times,\leq, (r)_{r\in  \RE}}$.
\end{proposition}
\begin{proof}
	The left-to-right direction is immediate as the constants $0$ and $1$ are definable in $\eso{\RE}{+,\times,\leq}$.
	For the converse direction, note that the numerical atom $\neg i\leq j$ is equivalent to the statement $j < i$. We show that $<$ is definable in $\peso{\RE}{+,\times,\leq,0,1}$. First note that every strictly positive real number $r\in\RE$ can be expressed by a ratio of two real numbers $n,m\in\RE$ such that $n,m\geq 1$. Moreover note that, for every such $n$ and $m$, the ratio $n/m > 0$. It is easy to see that the following $\peso{\RE}{+,\times,\leq,0,1} $-formula 
	\[
	\exists r \exists n \exists m ( 1\leq n \land 1\leq m \land n= r\times m \land i+r = j),
	\]
	where $r$, $n$, and $m$ are 0-ary function variables, expresses that $i < j$.
\end{proof}
Theorem \ref{thm:meer},
Proposition \ref{prop:firstcollapse}, Corollary \ref{cor:characterisation} together then yield the following:
\begin{corollary}
	$\PNP{}{\RE}= \NPr$ and $\PNP{0}{\RE} = \NPrz$.
\end{corollary}

The second natural question is: Are $\NP_{[0,1]}$ and
$\NP^0_{[0,1]}$ strictly below $\NPr$ and $\NPrz$?
 Again, the answer is no.
The proof of the following
proposition follows directly from the observation that arbitrary
real numbers can be encoded as ratios $x/(1-x)$, where $x\in [0,1]$, using an additional marker for sign. It is crucial to note that with negated numerical atoms one can express that the denominators of such encodings are positive; in the loose fragment this is not possible. The encodings needed can be clearly expressed in $\eso{[0,1]}{+,\times, \leq}$. We omit the proof.
\begin{proposition}
	$\eso{[0,1]}{+,\times, \leq, 0, 1}\equiv \eso{\RE}{+,\times,\leq, 0, 1}$ and  $\eso{[0,1]}{+,\times, \leq,(r)_{r\in \RE}}\equiv \eso{\RE}{+,\times,\leq,(r)_{r\in \RE}}$.
\end{proposition}
Hence Corollary \ref{cor:characterisation} yields the following:
\begin{corollary}
$\NP_{[0,1]}=\NPr$ and $\NP^0_{[0,1]}=\NPrz$.
\end{corollary}

Finally we consider a weakening of $\peso{\RE}{+,\times,\leq,0,1}$ by removing the constant $1$ from the language. It turns out that this small weakening has profound implications to the expressivity of the logic when restricted to function-free vocabularies.  

\begin{proposition}\label{prop:noone}
	Let $0\in S\subseteq \RE$. Then $\peso{S}{+,\times,\leq}
	\equiv \FO$ with respect
	to $\RE$-structures on function-free vocabularies.
\end{proposition}
\begin{proof}
	The direction $\FO\leq \peso{S}{+,\times,\leq}$ is self-evident. We give a proof for the converse. Let $\mA$ be an $\RE$-structure of  a function-free vocabulary $\tau$, $\phi\in \peso{S}{+,\times,\leq}[\tau]$  a formula, and $s$ an assignment for the first-order variables. Note that $\phi$ can be regarded also as a
formula of $\peso{\{0\}}{+,\times,\leq}$; we write $\phi_0$ to denote this interpretation. Let $\phi_\top$ denote the $\FO$-formula obtained from $\phi$ by removing the function quantifications in $\phi$ and replacing every numerical atom $i\leq j$ in $\phi$ with the formula $\exists x\, x=x$.
	Now note that there is a homomorphism from the first-order structure $(S, +, \times, \leq)$ to $(\{0\}, +, \times, \leq)$, and consequently,
	\(
	\mA \models_s \varphi \Leftrightarrow \mA \models_s \varphi_0.
	\)
	Here we note that $\varphi_0$ implies $\varphi$ since the second structure is a substructure of the first, and truth of existential formulae is preserved to extensions. Conversely, $\varphi$ implies $\varphi_0$ because atoms $i \leq j$ appear only positively, and the truth of formulae with only positive literals are preserved to homomorphic images.
	Since in the evaluation of $\phi_0$ every numerical term is evaluated to $0$ it follows that
	 \(
	 \mA \models_s \phi_0 \Leftrightarrow \mA \models_s \phi_\top.
	 \)
\end{proof}

\subsection{Separate branching on Boolean inputs and the existential theory of the reals}
It is known that on Boolean inputs $\NPrz$ coincides with the complexity class $\exists \RE$ (i.e., the class of problems polynomially reducible to the existential theory of the reals) \cite{BurgisserC06, SchaeferS17}. In this section we show an analogous result for $\PNP{0}{[0,1]}$.
\begin{definition}
	Define $\exists [0,1]^{\leq}$ to be the set of all languages $L\subseteq \{0,1\}^*$ for which there is a polynomial-time reduction $f$ from $\{0,1\}^*$ into sentences of existential loose $[0,1]$-guarded real arithmetic
	such that $x\in L$ iff $(\RE, +, \times, \leq,\allowbreak 0, 1)\models f(x)$.
\end{definition}
We show the following theorem:
\begin{theorem}\label{thm:E01}
	$\exists [0,1]^{\leq} = \bp{\PNP{0}{[0,1]}}$.
\end{theorem}
\begin{proof}{\,}
Note that the \textbf{\emph{right-to-left}} direction of this theorem follows immediately from Lemma \ref{new} by noting that the only real constants used by $\PNP{0}{[0,1]}$ S-BSS machines $M$ are $0$ and $1$, and that the Boolean inputs to $M$ can be defined in $\exists [0,1]^{\leq}$ by using the constants $0$ and $1$.

\textbf{\emph{Left-to-right.}}
	There exists a deterministic polynomial time Turing machine $M$ that given an input string computes the corresponding sentence $\phi$ of  existential loose $[0,1]$-guarded real arithmetic. Let $p$ be the polynomial that bounds the running time of $M$. Without loss of generality we may assume that, for any given input $i$ of length $n$, the formula computed by $M$ from input $i$ uses only variables $x_1,\dots,x_{p(n)}$.   
	Let $M^*$ be a nondeterministic S-BSS machine that, for a given input $i$ of length $n$, first guesses $p(n)$ many real numbers from the unit interval $[0,1]$ (these will correspond to the values of the variables $x_1,\dots,x_{p(n)}$). Then $M^*$ simulates the run of the deterministic polynomial time Turing machine $M$ on input $i$. Let $\phi$ be the formula computed this way. Finally we can use $M^*$ to check the matrix of $\phi$ using the values guessed for the variables $x_1,\dots,x_{p(n)}$. We omit further details, for the evaluation of the matrix can done essentially in the same way as in the \emph{left-to-right} direction of Theorem \ref{thm:main}.
\end{proof}

\section{Probabilistic team semantics}\label{sec:teamsemantics}
The purpose of this section is to characterise the descriptive complexity of probabilistic independence logic \cite{HKMV18}. 
 The formulae of this logic, and other logics that make use of dependency concepts involving quantities, are interpreted in probabilistic team semantics which generalises team semantics
by adding weights on variable assignments. A finite model together with a probabilistic team can then be seen as a particular metafinite structure, and thus a natural approach to computational complexity comes from BSS machines. 

Let $D$ be a finite set of first-order variables, $A$ a finite set, and $X$ a finite set of assignments (i.e., a \emph{team}) from $D$ to $A$. A $\emph{probabilistic team}$ $\X$ is then defined as a function 
\[\X\colon X\rightarrow [0,1]\]
such that $\sum_{s\in X}\X(s)=1$.
Also the empty function is considered a probabilistic team.
We call $D$ and $A$ the variable domain and value domain of $\X$, respectively. 

\emph{Probabilistic independence logic} ($\FO(\cpind)$) is defined as the extension of first-order logic with 
 probabilistic independence atoms $\pcixyz$ whose semantics is the standard semantics of conditional independence in probability distributions. Another probabilistic logic, $\FO(\approx)$, is obtained by extending first-order logic with \emph{marginal identity atoms} $\tuple x \approx \tuple y$ which state that the marginal distributions on $\tuple x$ and $\tuple y$ are identically distributed. The semantics for complex formulae are defined compositionally by generalising the team semantics of dependence logic to probabilistic teams. For details, not necessary in this paper, we refer the reader to \cite{HKMV18}. In principle, the point is that formulae of probabilistic independence logic define properties of $(\mA,\X)$ where $\A$ is a finite model and $\X$ a probabilistic team with value domain $\Dom(\mA)$. 
 
 \begin{example}
 Suppose we flip a coin. If we get heads, we roll two dice $x$ and $y$. If we get tails, we roll only $x$ and copy the same value for $y$. Repeating this procedure infinitely many times yields at the limit a probabilistic team (i.e., a joint probability distribution) over variables $x$ and $y$ satisfying
 \[(\pmi{x}{y}\vee x=y)\wedge \forall z\, x \approx z.\]
 By definition $\phi\vee \psi$ is true for a probabilistic team $\X$ if $\X$ is a mixture of two teams with respective properties $\phi$ and $\psi$ (here independence and (row-wise) identity between $x$ and $y$). By definition $\forall z \phi$ is true for a probabilistic team $\X$ if the extension of $\X$ with a uniform distribution for $z$ has the property $\phi$ (here identity between marginal distributions on $x$ and $z$).
 \end{example}
 
 We will now show that the descriptive complexity of probabilistic independence logic is exactly $\PNP{0}{[0,1]}$. For this we need some background definitions and results. 
 
\paragraph{Expressivity comparisons wrt. probabilistic team semantics}
Fix a relational vocabulary $\tau$. For a probabilistic team $\X$ with variable domain $\{x_1, \ldots ,x_n\}$ and value domain $A$, the function $f_\X:A^n\to [0,1]$ is defined as the probability distribution such that $f_\X(s(\tuple x))=\X(s)$ for all  $s\in X$. For a formula $\phi\in \FO(\cpind)$
of vocabulary $\tau$ and with free variables $\{x_1, \ldots ,x_n\}$, the class $\struc(\phi)$ is defined as the class of $\RE$-structures $\A$ over $\tau\cup\{f\}$ such that
 $(\A\upharpoonright \tau)\models_{\X} \phi$, where  $f_\X=f^\A$ and $\A\upharpoonright \tau$ is the finite $\tau$-structure underlying $\A$.

Let $\calL$ be any of the logics defined in Section \ref{sec:preli}.
We write $\FO(\cpind)\leq\calL$ if for every formula $\phi\in \FO(\cpind)$ of vocabulary $\tau$ there is a sentence $\psi\in \calL$ of vocabulary $\tau\cup\{f\}$ such that $\struc(\phi)= \struc^{d[0,1]}(\psi)$. Vice versa, we write $\calL \leq \FO(\cpind)$ if for every sentence $\psi\in \calL$ of vocabulary $\tau\cup\{f\}$ there is a formula $\phi\in \FO(\cpind)$ of vocabulary $\tau$ such that  $\struc(\phi)= \struc^{d[0,1]}(\psi)$.

\paragraph{Complexity characterisations wrt. probabilistic team semantics.}
 Let $\FO(\cpind)$ be a logic with vocabulary $\tau$ and $\calC$ a complexity class. Let $\calS$ be an arbitrary class of $\RE$-structures over $\tau\cup \{f\}$ that is closed under isomorphisms and where the interpretations of $f$ are distributions. We write $\enc(\calS)$ for the set of encodings of structures in $\calS$. Consider the following two conditions:
\begin{enumerate}[(i)]
	\item $\enc(\calS)=\{\enc(\A)\mid \A \in \struc(\phi)\}$ for some $\phi\in {\FO(\cpind)}\}$.
	\item $\enc(\calS)\in \calC$.
\end{enumerate}
If  $(i)$ implies $(ii)$, we write $\FO(\cpind) \leq \calC$, and if the vice versa holds, we write $\calC\leq \FO(\cpind)$.

 It is already known that probabilistic independence logic captures a variant of loose existential second-order logic in which function quantification ranges over distributions. This result was shown in two stages. First, it was proven in \cite{HKMV18} that the logic $\FO(\cpind,\approx)$
  is expressively equivalent to $\pesod{\SUM,\times,=}$.\footnote{In \cite{HKMV18} equi-expressivity with $\esod{\SUM,\times,=}$ is erroneously stated; the results in the paper actually entail equi-expressivity with $\pesod{\SUM,\times,=}$. 
}
 Later, it was proven in \cite{jelia19} that marginal identity can be expressed using independence, that is, $\FO(\cpind, \approx)$ is expressively equivalent to $\FO(\cpind)$.\footnote{In fact, $\FO(\cpind)$ is expressively equivalent to $\FO(\pind)$ which is the extension of first-order logic with \emph{marginal independence atoms} $\pmixy$, the semantics of which is the standard semantics of marginal independence in probability distributions \cite{jelia19}.}
\begin{theorem}[\cite{HKMV18,jelia19}]\label{thm:start}
$\FO(\cpind)\equiv \pesod{\SUM,\times,= }$.
\end{theorem}

We will now improve this result by removing the condition that restricts function quantification to distributions. For this we utilize a normal form lemma from \cite{HKMV18}. Observe that  we restrict attention to $d[0,1]$-structures, that is, all function symbols from the underlying vocabulary are interpreted as distributions.
\begin{lemma}[\cite{HKMV18}]\label{lemma}
For every $\pesod{\SUM,\times,=}$-formula $\phi$ there is an $\pesod{\SUM,\times,=}$-formula $\phi^*$ such that 
$\struc^{d[0,1]}{\phi}=\struc^{d[0,1]}{\phi^*}$, where $\phi^*$ is
of the form
$\exists \tuple f \forall  \tuple x\theta$,
where $\theta$ is quantifier-free and such that its second sort identity atoms are of the form $f_i(\tuple u,\tuple v) = f_j(\tuple u)\times f_k(\tuple v)$ or $f_i(\tuple u) = \SUM_{\tuple v}f_j(\tuple u,\tuple v)$ for
distinct $f_i,f_j,f_k$ such that at most one of them is not quantified.
\end{lemma}
\begin{lemma}\label{lem:esosummult2}
  $\pesod{\SUM,\times,=}\\\equiv_{d[0,1]}  \pesod{+,\times,=} \equiv_{d[0,1]}\pesou{+,\times,=,0,1}$.
\end{lemma}
\begin{proof}
We prove the claim in three steps, without relying on multiplication at any step. By  Proposition \ref{lem:esosub} we may assume that the finite domain is enriched with a
 successor function $S$ for tuples, its transitive derivatives $<, \leq$, and its minimal and maximal tuples $\tuple \min$ and $\tuple \max$ (of an appropriate arity),
  obtained by the lexicographic ordering induced from some linear ordering $\fleq$. Additionally, we may assume a constant $c$ on the finite domain.

\noindent 
\textbf{Step 1:} $\pesod{\SUM,\times,=} \leq_{d[0,1]} \pesod{+,\times,=}$. 
We may assume that any $\pesod{\SUM,\times,=}$ formula is of the form
 stated in Lemma \ref{lemma}. 
 Thus it suffices to express in $\pesod{+,\times,=}$ each numerical identity of the form $f(\tuple u)=\SUM_{\tuple x}f'(\tuple u,\tuple x)$.
 First, we quantify a $2m$-ary distribution variable $g$ upon which we impose:
\begin{align}\label{partial}
\forall \tuple x\tuple y\big[&g(\tuple x,\tuple \min)+g(\tuple x,\tuple \min)=f'(\tuple u,\tuple x) \wedge\\
& \big( \tuple y< \tuple \max\to\nonumber\\
& g(S(\tuple y),S(\tuple y))+g(S(\tuple y),S(\tuple y))= g(S(\tuple y),\tuple y)+g(\tuple y,\tuple y)\big)\wedge \nonumber\\
& \big(  S(\tuple y)<\tuple x \to\nonumber\\
 & g(\tuple x,S(\tuple y))+g(\tuple x,S(\tuple y)) = g(\tuple x,\tuple y)\big)\big].\nonumber
\end{align}
The point
is to calculate partial sums $\SUM_{\tuple x\leq y}f'(\tuple u,\tuple x)$ and store sufficiently small fractions of them in $g(\tuple y,\tuple y)$.
 Suppose $\tuple y$ is the $n$th tuple.
 Then 
\[g(\tuple y,\tuple y)=\frac{1}{2^{n}}(f'(\tuple u,\tuple \min)+ \ldots +f'(\tuple u,\tuple y)),\] 
and for 
$\tuple x > \tuple y$,
\[g(\tuple x,\tuple y)=\frac{1}{2^{n}} f'(\tuple u,\tuple x).\] 
Consequently, the sum of all $g(\tuple x,\tuple y)$ where $\tuple x \geq \tuple y$ is at most $1$. By allocating the remaining weights to $(\tuple x,\tuple y)$ such that $\tuple x < \tuple y$, it follows that $g$ is a distribution.

Furthermore, we quantify a $2m$-ary distribution variable $h$  satisfying:
\begin{align*}
\forall \tuple x &[  h(\tuple \min)+h(\tuple \min)= f(\tuple u)\wedge \\
&  \tuple x < \tuple \max \to h(S(\tuple x)) + h(S(\tuple x))= h(\tuple x)]. 
\end{align*} 
It follows that $
h(\tuple y)= \frac{1}{2^{n}} f(\tuple u)$. Consequently, $g(\tuple \max,\tuple \max)= h(\tuple \max)$ if and only if $f(\tuple u)=\SUM_{\tuple x}f'(\tuple u,\tuple x)$. Note that $h$ is not a distribution since the weights do not add up to $1$. However, we may increment the arity of $h$ by one and replace $h(\tuple x)$ above with $h(\tuple x,c)$.
 Then $h$ is a distribution if the remaining weights are pushed to $h(\tuple x,y)$, where $y\neq c$. This concludes the proof of Step 1.

\noindent
\textbf{Step 2:} We show a stronger claim: $\pesod{+,\times,=}\le\pesou{+,\times,=,0,1}$. For this, it suffices to show how to 
 express in $\pesou{+,=,0,1}$ that a function $f$
  is a distribution. The following formula expresses just that:
\begin{align*}\exists g \big(&g(\tuple \min)=f(\tuple \min) \wedge \\
&\forall \tuple x ( \tuple x < \tuple \max  \to g( S(\tuple x))= g(\tuple x)+f(S(\tuple x))) \wedge g(\tuple \max) = 1\big).
\end{align*}

\noindent
\textbf{Step 3:} We show a stronger claim: $\pesou{+,\times,=,0,1}\\\leq_{[0,1]} \pesod{\SUM,\times,=}$. Suppose $\phi$ is some formula in $\pesou{+,\times,=,0,1}$.
 Let $k$ be the maximal arity of any function variable/symbol appearing in $\phi$, and suppose $n$ is the size of the finite domain; the total sum of the weights of a function
 is thus at most $n^k$.
We now show how to obtain from $\phi$ an equivalent formula in $\pesod{\SUM,\times,=}$; the idea is to scale all function weights by $1/n^k$. We have two cases:

\noindent
\textit{Function variables.}
If $f$ is an $m$-ary quantified function variable, we replace it with
 an $(m+1)$-ary quantified distribution variable $d_f$ satisfying
\begin{equation*}\label{eq:distsim}
\forall \tuple x \exists d' \forall \tuple y \,d'(\tuple y,c) = d_f(\tuple x,c),
\end{equation*}
 where $d'$ is a $(k+1)$-ary distribution variable.
 Now $n^k d_f(\tuple x,c)\le 1$ because $d'$ is a distribution, and thus $ d_f(\tuple x,c)\leq\frac{1}{n^k}$.
 
 \noindent
\textit{Function symbols.}
Suppose $f(\tuple x)$ is a function term which appears as a term or subterm in $\phi$, and $f$ is a function symbol from the underlying vocabulary. We quantify a $(k+1)$-ary distribution variable $d_{f(\tuple x)}$ satisfying
\begin{equation*}
\forall \tuple x(\SUM_{\tuple y}d_{f(\tuple x)}(\tuple y,c)=f(\tuple x) \wedge\forall \tuple y\tuple z d_{f(\tuple x)}(\tuple y,c)=d_{f(\tuple x)}(\tuple z,c)).
\end{equation*}
 It follows that $d_{f(\tuple x)}(\tuple x,c)=\frac{1}{n^k}f(\tuple x)$. Since  $f(\tuple x)\le 1$, we may define $d_{f(\tuple x)}$ as a distribution.
  
 Observe now that each numerical atom appearing in $\phi$ is an identity between two multivariate polynomials over function terms. Without loss of generality all the constituent monomials in these atoms are of a fixed degree $D$ and have coefficient one;  note that each monomial with degree less than $D$ can be appended in $\pesou{+,\times,=,0,1}$ with a quantified nullary function $n$ taking value $1$. We now replace in each numerical atom $i=j$ function terms $f(\tuple x)$ with $d_f(\tuple x,c)$ or $d_{f(\tuple x)}(\tuple x,c)$, depending on whether $f$ is a function variable or a function symbol. Thus we represent $i=j$ in $\pesod{\SUM,\times,=}$ as $\frac{i}{n^{Dk}}=\frac{j}{ n^{Dk}}$, wherefore not only its truth value, but also that of $\phi$, is preserved in the transformation.
\end{proof}
By combining Corollary \ref{cor:characterisation}.\ref{this}, Theorem \ref{thm:start}, and Lemma \ref{lem:esosummult2}, we finally obtain the following result.
\begin{theorem}
$\FO(\cpind)\equiv \PNP{0}{[0,1]}$.
\end{theorem}

\section{Concluding remarks}\label{sec:conclusion}
Applications of logic in AI and advanced data management require
probabilistic interpretations, a role that is well fulfilled by
probabilistic team semantics.  On the other hand, in the theory
of computation and automated reasoning, computation and logics
over the reals are well established with solid foundations.  In
this paper we have provided bridges between the two worlds. 
We introduced a novel variant of BSS machines and provided a logical and topological characterisation of its computational power. In addition, we determined the expressivity of probabilistic independence logic with respect to the BSS model of computation.

There are many interesting directions of future research. One is to consider the additive fragment of BSS computation. Restricted to Boolean inputs it is known that, if unrestricted use of machine constants is allowed, the additive $\NPr$ branching on equality collapses to $\NP$ and branching on inequality captures $\NP/poly$ \cite{Koiran94}. What can we say about the additive fragment of S-BSS computation? Another direction is to devise logics that characterise other important complexity classes over S-BSS machines. Gr\"adel and Meer \cite{GradelM95} established a characterisation of polynomial time on ranked $\RE$-structures using a variant of least fixed point logic. In the setting of team semantics and classical computation, Galliani and Hella \cite{gallhella13} showed that the so-called \emph{inclusion logic} characterises polynomial time on ordered structures. Can we extend the applicability of these results to the realms of S-BSS computation and probabilistic team semantics? Finally, we would like to devise natural complete problems for the complexity classes defined by S-BSS machines. In particular, we would like to obtain a natural complete problem for $\exists [0,1]^\leq$; a weakening of the art gallery problem is one promising candidate.
We conclude with a few open problems:
\begin{itemize}
	\item Is $\exists[0,1]^{\leq}$ strictly included in $\exists \RE$? A positive answer would be a major breakthrough, as it would separate $\NP$ from $\PSPACE$.
	\item We know that $\NP\leq \exists [0,1]^\leq \leq \exists \RE \leq \PSPACE$. Can we establish a better upper bound for $\exists [0,1]^\leq$? In particular, is $\exists [0,1]^\leq$ contained in the polynomial hierarchy?
	\item We established that S-BSS computable languages are included in the class of BSS computable languages that are countable disjoint unions of closed sets. Does the converse hold?
\end{itemize}

\begin{acks} 
The first and the second author were supported by the Academy of Finland grant 308712. The third and the fourth author were supported by the Research Foundation Flanders grant
G0G6516N.
 The third author was partially
supported by the National Natural Science Foundation of China under grant 61972455, and
 the fourth author was
 an international research fellow
 of the Japan Society for the Promotion of Science, Postdoctoral Fellowships for Research in Japan (Standard).
\end{acks}

 \bibliographystyle{ACM-Reference-Format}
\bibliography{biblio}

\end{document}